\documentclass[final,seceqn]{elsart1p}
\usepackage{graphicx}
\usepackage{amsmath,amssymb,amsfonts,bm}
\journal{Nuclear Physics B}
 \newtheorem{theorem}{Theorem}[section]

 \newenvironment{proof}[1][Proof]{\begin{trivlist}
 \item[\hskip \labelsep {\bfseries #1}]}{\end{trivlist}}

 \newcommand{\Or}{\mathrm{O}}
 \renewcommand{\Re}{\mathrm{Re}}
 \renewcommand{\Im}{\mathrm{Im}}

\begin{document}

\begin{frontmatter}

\title{On conjectured local generalizations of anisotropic scale
  invariance and their implications}
\author[UDE,Minsk]{S. Rutkevich},
\author[UDE]{H. W. Diehl},
\author[UDE,Lviv]{M.~A.\ Shpot}

\address[UDE]{Fakult\"at f\"ur Physik, Universit{\"a}t Duisburg-Essen,
D-47048 Duisburg, Germany}
\address[Minsk]{Institute of Solid State and
   Semiconductor Physics,
 220072 Minsk, Belarus\thanksref{ppad}}
\address[Lviv]{Institute for Condensed Matter Physics, 79011 Lviv, Ukraine\thanksref{ppad}}
\thanks[ppad]{Present and permanent address}
\begin{abstract}
The theory of generalized local scale invariance of strongly anisotropic scale invariant systems proposed some time ago by Henkel [Nucl.\ Phys.\ B \textbf{641}, 405 (2002)] is examined.  The case of so-called type-I systems is considered. This was conjectured to be realized by systems at $m$-axial Lifshitz points; in support of this
claim, scaling functions of two-point cumulants at the uniaxial Lifshitz point of the three-dimensional ANNNI model were predicted on the basis of this theory and found to be in excellent agreement with Monte Carlo results [Phys.\ Rev.\ Lett.\ \textbf{87}, 125702 (2001)]. The consequences of the conjectured invariance equations are investigated, with emphasis put on this case. It is shown that fewer solutions than anticipated by Henkel generally exist and contribute to the scaling functions if these equations are assumed to hold for all (positive and negative) values of the $d$-dimensional space (or space time) coordinates $(t,\bm{r})\in \mathbb{R}\times\mathbb{R}^{d-1}$. Specifically, a single rather than two independent physically acceptable solutions exists in the case relevant for the mentioned fit of Monte Carlo data for the ANNNI model. Renormalization-group improved perturbation theory in $4+m/2-\epsilon$ dimensions is used to determine the scaling functions of the order-parameter and energy-density two-point cumulants in momentum space to two-loop order.  The results are mathematically incompatible with Henkel's predictions except in free-field-theory cases. However, the scaling function of the energy-density cumulant we obtain for $m=1$ upon extrapolation of our two-loop RG results to $d=3$ differs numerically little from that of an effective free field theory.
\end{abstract}

\end{frontmatter}
\tableofcontents
 \bibliographystyle{elsart-num}

% main text
\section{Introduction}
  \label{sec:intro}
It is a fact, well established by plenty of experiments and theoretical works, that the large-scale physics of systems at critical points can be described by scale-invariant continuum field theories \cite{DG76,Fis83,ZJ96}. In many cases the associated probability distributions are besides scale invariant also translation and rotation invariant when expressed in appropriate variables.  During the past 25 years it has  become widely appreciated that in those cases where long-range interactions are absent or may be ignored, such invariance under translations, rotations, and scale transformations usually entails the invariance under a larger symmetry group --- that of conformal transformations \cite{CCJ70,Car87,Car90,Gin90,DiFMS97}.

The application of conformal invariance to critical phenomena started long ago with a short note by Polyakov \cite{Pol70}. More detailed investigations \cite{WG74,Sch76} followed soon, but the issue did not receive much attention until the seminal work of Belavin et al \cite{BPZ84a,BPZ84} on two-dimensional conformal field theories. This revealed the enormous potential of conformal invariance, triggering an outburst of research activities, which in turn led to impressive success in many other fields, such as bulk, finite-size \cite{Bar83,ISZ88,Hen99}, and boundary critical behavior \cite{Car87,Car90,Bin83,Die86a,Car84,Die97}, polymer physics \cite{Eis93}, quantum impurity problems \cite{FLS95}, and string theories \cite{Gin90,DiFMS97,Kak91}.

Conformal transformations locally correspond to combinations of translations, rotations, and scale transformations with a position-dependent scale factor $\ell(\bm{x})$ that involve no shear. Thus, conformal invariance may be viewed as the generalization of a global symmetry to a local one.

There exists a wealth of systems in nature that exhibit global scale invariance of a distinct and more general kind, called \emph{anisotropic scale invariance} (ASI). Its characteristic feature is that an anisotropic rescaling of the space separations (or spacetime separations in time-dependent phenomena) along various axes by at least two (or more) distinct powers of a scale factor $\ell$ is required to make such systems statistically self-similar. Familiar examples are uniaxial dipolar ferro- and antiferromagnets at their critical points (see, e.g., Refs.~\cite{Aha76} and  \cite[chapter~27.5]{ZJ96}), systems at Lifshitz points \cite{Hor80,Sel92,Die02}, and dynamical critical phenomena near and away from thermal equilibrium \cite{HH77,SZ95}, among them driven diffusive systems \cite{SZ95}, stochastic surface-growth processes \cite{Kru97}, directed percolation, and spreading processes \cite{Hin00,Odo04}.

In the case of static critical behavior at an $m$-axial Lifshitz point (LP) in $d$ space dimensions, the position vector $\bm{x}=(x_1,\dotsc,x_d)$ of Cartesian coordinates $x_\gamma$ can be decomposed as $\bm{x}=(\bm{z},\bm{r})$ into the $m$- and ($d-m$)-dimensional components $\bm{z}=(x_\alpha)$ and $\bm{r}=(x_\beta)$ with $\alpha=1,\dotsc,m$ and $\beta=m+1,\dotsc,d$, respectively. ASI is encoded in the transformation property
\begin{equation}
    \label{eq:ASIO}
  \mathcal{O}_j(\ell^\theta \bm{z},\ell\,\bm{r})
=\ell^{-\Delta_j}\,\mathcal{O}_j(\bm{z},\bm{r})
\end{equation}
of local scaling operators $\mathcal{O}_j(\bm{x})$ with scaling dimensions $\Delta_j$, where $\theta$, the anisotropy exponent, differs from $1$. The obvious analog for time-dependent phenomena involving an isotropic rescaling of distances but distinct rescaling of time reads
\begin{equation}
   \label{eq:ASIOst}
   \mathcal{O}_j(\ell^{\mathfrak{z}} t,\ell\,
   \bm{r})=\ell^{-\Delta_j}\,\mathcal{O}_j(t,\bm{r})\;,
\end{equation}
where $\mathfrak{z}$ is the so-called dynamic critical exponent. As a consequence of these properties, the multi-point correlation functions of such operators take scaling forms. Consider, for example, the case of Eq.~(\ref{eq:ASIOst}), and assume that the systems are translation invariant in both time and space, as well as rotation invariant. Together with the presumed ASI, these properties imply that the two-point cumulant function of two such operators $\mathcal{O}_j$ and $\mathcal{O}_k$ can be written as
\begin{equation}
   \label{eq:CorSc}
\langle \mathcal{O}_j( t,\bm{ r})\,\mathcal{O}_k(t',\bm{ r}')\rangle^{\text{cum}}=
|\bm{ r}-\bm{ r}'|^{-\Delta_j-\Delta_k}\;
\Omega_{jk}{\left(\frac{t-t'}{|\bm{ r}-\bm{ r}'|^\mathfrak{z}}\right)}\;.
\end{equation}

Given the enormous success the use of conformal invariance has had in the study of isotropic critical behavior, a natural question to ask is whether global ASI, in conjunction with appropriate other global symmetries, such as translation and rotation invariance, would again entail more powerful local symmetries that impose useful constraints on the scaling functions or even determine them completely.

This idea has been pursued for many years, in particular, by Henkel who proposed a phenomenological approach termed ``local scale invariance (LSI)'' in a series of papers \cite{Hen99,Hen97,Hen02} and applied it to a variety of systems exhibiting ASI. Postulating a set of ``axioms of local scale invariance'', he suggested that the two-point scaling functions $\Omega(u)$ of various systems exhibiting ASI should satisfy differential equations.  According to him there should be two classes of local generalizations of ASI: The first, denoted type I, should apply to anisotropic scale-invariant equilibrium systems; the second, type II, to time-dependent scale-invariant phenomena. As a nontrivial example of type II, the relaxational behavior of systems representing the dynamic universality class of the so-called stochastic model A \cite{HH77}, following a quench from an initial disordered state to the critical point, was suggested. Subsequent analytical calculations based on the $\epsilon=4-d$ expansion for model A \cite{CG02a,CG02b} yielded definite, albeit small, violations of Henkel's predictions at two-loop order.

As nontrivial realizations of his type I of generalized ASI, Henkel suggested  equilibrium systems at $m$-axial LP. To check the predictions of his theory, Pleimling and him performed extensive Monte Carlo simulations \cite{PH01} for the two-point correlation function of the $d=3$ dimensional axial next-nearest-neighbor Ising (ANNNI) model \cite{Hor80,Sel92,Die02} at its uniaxial LP.  In their original application of the phenomenological LSI approach to this problem, they assumed that the anisotropy exponent $\theta$ takes its classical value $\theta=1/2$. Their theory then predicted the scaling function to be a linear combination of two linearly independent solutions of a differential equation, involving a single free parameter which they determined from Monte Carlo data for moments. The so-obtained scaling function appeared to be in perfect agreement with their Monte Carlo data. However, the $\epsilon$ expansion about the upper critical dimension $d^*(m)=4+m/2$ yields deviations of $\theta$ from its classical value $1/2$ at order $\epsilon^2$ \cite{DS00a,SD01,SPD05,SDP08}.%
\footnote{The recently developed large-$n$ expansion for the study of  critical behavior at LP, where $n$ is the number of   components of the order parameter, also gives nonclassical values of   $\theta$ for $d<d^*(m)$ \cite{SPD05,SDP08}.}
Pad\'e estimates based on these series expansions to $\mathrm{O}(\epsilon^2)$ gave $\theta\simeq 0.47$ for the uniaxial case $m=1$ at $d=3$. To account for such nonclassical values of $\theta$, Pleimling and Henkel \cite{rem:thetagen} generalized their LSI predictions for the scaling function by expanding in $\theta-1/2$. They found that the resulting predictions remained in agreement with their Monte Carlo data provided a value for $\theta$ sufficiently close to $1/2$ ($0.47\lesssim\theta\lesssim 0.5$) was chosen.

Unlike the case of type II, Henkel's predictions obtained via his phenomenological LSI approach have not yet been checked in a systematic fashion by mathematically well controlled analytical calculations. The only exceptions we are aware of are mean spherical models. Their propagators at the LP are those of massless free field theories. LSI does not lead to new nontrivial consequences for them. Hence they are unsuitable for critical checks of the predictive power and viability of this approach. In view of the apparent excellent agreement of the Monte Carlo data of Ref.~\cite{PH01} with the scaling function obtained by the LSI approach we feel that nontrivial checks of this approach through analytical calculations for nontrivial models of type-I systems are urgently needed.

The aim of this paper is to perform such checks. To this end we shall investigate standard $n$-component $\phi^4$ models for the description of critical behavior at $m$-axial LP, use the $\epsilon=4+m/2-d$ expansion to compute appropriate two-point scaling functions, and compare the results with the predictions of Henkel's LSI approach. For the sake of simplicity, we shall focus our attention in most of our work on the uniaxial case $m=1$. We begin in Section~\ref{sec:revLSI} with a brief review of the main predictions of this theory for the scaling functions of such type-I systems.  In Section~\ref{sec:rv} we first discuss the region of validity of the suggested invariance equations of the two-point correlation function in position space. We then transform these equations to momentum space, and discuss the consequences for the scaling forms of the Fourier transformed two-point functions. In Section \ref{sec:model} we introduce the standard continuum model representing the universality class of critical behavior at $m$-axial LP in $d$ dimensions. We then investigate the $\epsilon$ expansion  for the energy-density correlation function, using dimensional regularization in conjunction with minimal subtraction of ultraviolet (UV) poles. In Section \ref{sec:corrfctdm3} we focus on the analytically tractable cases $d=m+3$ and $m=2$. The former has the simplifying feature that the scaling function of the free propagator  in position space reduces at the LP to a Gaussian. This allows us to obtain explicit expressions for the 
energy-density and order-parameter correlation functions to order $\epsilon^2$ and cast them in scaling form. In Section~\ref{sec:endens} we determine the two-point correlation function of the energy density for the uniaxial case $m=1$ to two-loop order. A detailed comparison of our results for this and the previously mentioned correlation functions at the LP with the predictions of Henkel's phenomenological theory follows in Section \ref{sec:concl}. It shows that these predictions do not hold except in the trivial case of a Gaussian LP. Whenever loop corrections to the correlation functions cannot be neglected in the interacting case, the predicted scaling functions are inconsistent with our findings.  Finally, there are four appendices to which we have relegated various computational details.

\section{Conjectured properties of scaling functions}
\label{sec:revLSI}
Our objective is to check the predictions of the  LSI theory  proposed in Refs.~\cite{Hen99,Hen97,Hen02} for strongly anisotropic critical systems of type I. We begin by  recalling the basic postulates on which this  theory is based and its conjectured properties of scaling functions.

 To this end, we consider the pair correlation functions of two quasiprimary scaling operators $\mathcal{O}_j(\bm{z}_j, \bm{r}_j)$, $j=1,2$,  with zero averages $\langle\mathcal{O}_j\rangle$ and the behavior~\eqref{eq:ASIO} under global scale transformations.  For the sake of notational simplicity, we focus on the uniaxial case $m=1$. In order to facilitate comparisons with Henkel's work, we shall follow him and denote the one-dimensional equivalent of the variables $\bm{z}_j$ by $t_j$. We assume translation invariance in $t$~space as well as translation and rotation invariance in $\bm{r}$~space, define the scaling dimension
 \begin{equation}
 \label{eq:Delta}
\Delta\equiv(\Delta_1+\Delta_2)/2,
\end{equation}
and introduce the notations $t=t_1-t_2$, $\bm{r}=\bm{r}_1-\bm{r}_2$, and $r=|\bm{r}|$. By analogy with Eq.~\eqref{eq:CorSc}, the pair correlation functions can be written as
 \begin{equation}
 \label{eq:Gdef}
G_{jk}(t,r)
\equiv\langle\mathcal{O}_j(t_1,\bm{r}_1)\,\mathcal{O}_k(t_2,\bm{r}_2)\rangle=r^{-2\Delta}\,\Omega_{jk}(t\,r^{-\theta}).
\end{equation}
Their scaling form reflects the invariance under global anisotropic scale transformation generated by
\begin{equation}
 \label{eq:X0def}
 X_0= -t\,\partial_t - \frac{1}{\theta} \,\bm{r}\cdot\partial_{\bm{r}} - \frac{2 \Delta}{\theta}.
\end{equation}
Hence we have
\begin{equation}
 \label{eq:XgenG}
 X_0\,G_{jk}(t,r)=0.
\end{equation}

To proceed it will be helpful to recall some essentials of Henkel's approach \cite{Hen97,Hen02} without going into details. His starting point is the well-known algebra associated with Schr\"odinger invariance. This he generalizes by allowing for values of $\theta\ne 2$ and anomalous dimensions of scalar quasiprimary fields. He then imposed the requirement that the generators yield a finite number of independent conditions when applied to the two-point functions of quasiprimary fields. Exploiting the consequences, he was able to identify two distinct classes of systems, called type I and type II, respectively. For the type-I systems with which we are concerned here, the anisotropy exponent $\theta$ is constrained to the fractional values
\begin{equation}
 \label{eq:thetasdef}
\theta\equiv\theta_N=2/N,\;\;N\in\mathbb{N},
\end{equation}
so that Eq.~\eqref{eq:XgenG} simplifies to
\begin{equation}
 \label{eq:X0G}
  X_0G_{jk}(t,r)=\left( -t\partial_t - \frac{N}{2} \,\bm{r}\cdot\partial_{\bm{r}} - N \Delta\right)G_{jk}(t,r)=0.
\end{equation}

The other assumptions of Henkel are that the $G_{jk}$ are also annihilated by generators denoted as $\bm{Y}_{1-N/2}=(Y_1,\ldots, Y_{d-1})$ and $X_1$ defined via
\begin{equation}
 \label{eq:YG}
 \bm{Y}_{1-N/2} G_{jk}(t,r)\equiv\bigg( - t \partial_{\bm{r}} -
 \frac{2\alpha_1}{N} {\bm r} \partial_t^{N-1}
 \bigg) G_{jk}(t,r) = 0
\end{equation}
and
 \begin{eqnarray}\label{eq:X1G}
 X_1G_{jk}(t,r)&\equiv&( -t^2\partial_t - N t \,\bm{r}\cdot \partial_{\bm{r}} - 2 N \Delta_1 t
 -  \alpha_1 {\bm r}^2 \partial_t^{N-1} ) G_{jk}(t,r)\nonumber\\ &&\strut +2t_2 X_0 G_{jk}(t,r)
  +N\, {\bm r}_2 \cdot\bm{Y}_{1-N/2} G_{jk}(t,r) = 0
  \end{eqnarray}
 where $\alpha_1$ is a nonzero parameter. 
 
Note that when $N$ is taken to be an arbitrary real number so that the condition~\eqref{eq:thetasdef} is not satisfied, the generators $ \bm{Y}_{1-N/2}$ and $X_1$ involve fractional derivatives. Since definitions of fractional derivatives $\partial_t^\iota$ other than via  Fourier transformation $\partial_t^\iota\leftrightarrow (ik)^\iota$ are in use, the precise definition of these fractional derivatives becomes an issue. Background on this matter and Henkel's choice of their definition can be found in reference \cite[Appendix A]{Hen02}. We shall exclusively have to deal with the above equations in those cases where condition~\eqref{eq:thetasdef} is satisfied. All derivatives then reduce to conventional partial derivatives of first and higher orders. Clearly, any acceptable definition of fractional derivatives $\partial_t^\iota$ must reduce to such standard derivatives for nonnegative integer values of $\iota$, i.e., when $N$ becomes a natural number.  Henkel's choice indeed fulfills this condition. We therefore do not have to worry about potential differences resulting from distinct definitions of fractional derivatives here and in the following.

The meaning of the first condition, Eq.~\eqref{eq:X0G}, has already been explained. The second condition, Eq.~\eqref{eq:YG}, reduces in the special cases $N=1$ and $N=2$ to familiar ones implied by  invariance  under global projective Galilei transformations and rotations, respectively. The third one, equation (\ref{eq:X1G}), is  reminiscent of the one that follows for systems exhibiting  isotropic scale invariance in $\bm{x}\equiv (t,\bm{r})$ space from the invariance under M{\"o}bius transformations. As discussed in Ref.~\cite[p.~430]{Hen02}, the three conditions~\eqref{eq:X0G}--\eqref{eq:X1G} can be combined to obtain the constraint
\begin{equation}
 \label{eq:Del1eq2}
\Delta_1=\Delta_2,
\end{equation}
unless $G_{jk}\equiv0$. One can therefore put $\Delta_1=\Delta$ and drop the subscript $1$ on both $\Delta_1$ and $\alpha_1$.%
\footnote{Reference~\cite{Hen02} also uses a parameter $\alpha_2$. However, this is related to $\alpha_1$ via $\alpha_2=(-1)^{-N}\alpha_1$ in the case of type-I systems we are concerned with here.}
Both this constraint and Eq.~\eqref{eq:X0G} are satisfied by the scaling ansatz
\begin{equation}
 \label{eq:Gscfs}
G_{jk}(t,r)=\delta_{\Delta_j,\Delta_k}\,G(t,r),\quad G(t,r)=r^{-2\Delta}\,\Omega^{(N)}(tr^{-2/N}).
\end{equation}
Its substitution into Eq.~\eqref{eq:YG} then yields a differential equation for the scaling function, namely
 \begin{equation}
  \label{eq:Hv}
 \left( \alpha\,\frac{d^{N-1}}{dv^{N-1}}-v^2 \frac{d}{dv}-\zeta v\right)\Omega^{(N)}(v)=0\; \mbox{ with }\zeta=N\Delta.
 \end{equation}

 Henkel considers this equation on the interval $[0,\infty)$ subject to the boundary  conditions
 \begin{eqnarray}
  \label{eq:Bc0}
 \lim_{v\to 0}\Omega^{(N)}(v)&=&\Omega^{(N)}(0) \equiv \omega^{(N)}_0, \\
 \Omega^{(N)}(v)& \mathop{\approx}_{v\to\infty}&\omega^{(N)}_{\infty}\, v^{-\zeta},\label{eq:Bcinf}
 \end{eqnarray}
 where $\omega^{(N)}_0\ne0$ and $\omega^{(N)}_\infty$ are constants. Assuming that $ N\ge 2$, he arrives at the general solutions
 \begin{equation}
   \label{eq:hyp}
  \Omega^{(N)}(v)=\sum_{p=0}^{N-2} b^{(N)}_p\, v^p \mathcal{F}_p(v)
  \end{equation}
  with
  \begin{equation}
   \label{eq:Fp}
   \mathcal{F}_p(v)={}_2F_{N-1}\bigg(\frac{\zeta+p}{N},1;
   1+\frac{p}{N},1+\frac{p-1}{N},\ldots,
 \frac{p+2}{N};\frac{v^N}{N^{N-2}\,\alpha}\bigg),
 \end{equation}
 where  ${}_2F_{N-1}$ is the generalized hypergeometric function, while $b^{(N)}_p$ are  free parameters.

Using known theorems \cite{Wri40} about the asymptotic behavior of the functions ${}_2F_{N-1}(x)$ in the limit $x\to\infty$, he finds that the right-hand side of the solutions~\eqref{eq:hyp}, for general values of $b^{(N)}_p$, would diverge asymptotically as
 \begin{equation}
 \label{eq:ex}
  \Omega^{(N)}(v)\mathop{\sim}_{v\to\infty}\big(v\alpha^{-1/N}\big)^{(\zeta+1-N)/(N-2)}
 \exp\Big[\frac{N-2}{N}\,\big(v\alpha^{-1/N}\big)^{N/(N-2)}\Big]
 \end{equation}
 in the large-$v$ limit, where the proportionality constant is a linear combination of the coefficients $b^{(N)}_p$. Since such behavior is inconsistent with the boundary condition~(\ref{eq:Bcinf}), he requires that this constant vanishes. This implies the condition
 \begin{equation}
   \label{eq:linb}
   \sum_{p=0}^{N-2}b^{(N)}_p\,\frac{\Gamma(p+1)}{\Gamma((p+1)/N)\Gamma((p+\zeta)/N)}\Big(\frac{\alpha}{N^2}\Big)^{p/N}=0,
 \end{equation}
which can be used to  eliminate the coefficient $b^{(N)}_{N-2}$. As a consequence, the solutions~(\ref{eq:hyp}) become
\begin{equation}
 \label{eq:Omegaform}
\Omega^{(N)}(v)=\sum_{p=0}^{N-3} b^{(N)}_pv^p\, \Omega^{(N)}_p(v),
\end{equation}
with
 \begin{eqnarray}
   \label{eq:hyp1}
  \Omega^{(N)}_p(v)&=v^p\,\mathcal{F}_p(v)&-
 \frac{p!\,\Gamma[(N-1)/N]\,\Gamma[1+(\zeta-2)/N]}{(N-2)!\,\Gamma[(p+1)/N]\,\Gamma[(p+\zeta)/N] } \nonumber\\
 &&\strut\times
  \big(\alpha/N^2 \big)^{(p+2-N)/N}\,v^{N-2} \mathcal{F}_{N-2}(v).
 \end{eqnarray}

Condition~\eqref{eq:linb} ensures the cancellation of  the leading exponentially diverging terms in the limit $v\to+\infty$ of $\Omega^{(N)}(v)$. In order to comply with the boundary condition~\eqref{eq:Bcinf}, no other diverging or non-decaying terms would have to remain in $\Omega^{(N)}(v)$ in the limit $v\to\infty$. Provided this is the case, the general solution of  equation (\ref{eq:Hv}) subject to the  boundary conditions~\eqref{eq:Bc0} and \eqref{eq:Bcinf} involves $N-2$  free parameters  $b^{(N)}_p$ with $p=0,\ldots,N-3$ and is given by Eqs.~\eqref{eq:Omegaform} and \eqref{eq:hyp1}. That the boundary condition~\eqref{eq:Bcinf} is satisfied was confirmed in Ref.~\cite{Hen02}  by numerical means for $N=4,5,6$. 

In Appendix~\ref{app:Om} we reconsider in detail the problem of solving Eq.~\eqref{eq:Hv} subject to the boundary conditions~\eqref{eq:Bc0} and \eqref{eq:Bcinf}. We prove there the following statements about  the solutions given by Eqs.~\eqref{eq:Omegaform} and \eqref{eq:hyp1} with general values of the coefficients $b^{(N)}_p$:  When $N=3,4,5$, they comply indeed with the boundary condition~\eqref{eq:Bcinf}. For general integer values $N\ge7$, this boundary condition gets violated by the presence of exponentially diverging terms in the large-$v$ limit. Requiring the absence of these (subleading) divergences imposes further restrictions on the coefficients $b^{(N)}_p$, which reduce their number. For example, for $N=7$, the general solution of Eq.~\eqref{eq:Hv} satisfying the boundary conditions~\eqref{eq:Bc0} and \eqref{eq:Bcinf} involves only $3$ rather than $5$ free parameters. The case $N=6$ is special. As we show in Appendix~\ref{app:Om}, the scaling function $\Omega^{(6)}$ suggested by Henkel (for general values of the coefficients $b^{(6)}_p$, $p=0,\ldots,3$) violates again the boundary condition~\eqref{eq:Bcinf} but diverges only algebraically as $v\to\infty$. Requiring the absence of this divergence reduces the number of free parameters to $3$.

The above solutions $\Omega^{(N)}(v)$ for $N=4$ were used in Refs.~\cite{Hen99}, \cite{Hen97}, \cite{Hen02} and \cite{PH01} as predictions for the scaling function of the order-parameter pair correlation function of three-dimensional systems at Lifshitz points. Furthermore, in Ref.~\cite{PH01} extensive Monte Carlo data were presented for the scaling function of the three-dimensional ANNNI model, which appeared to be in perfect agreement with these predictions.  Since this case $N=4$ is of particular interest to us, we give here the explicit form of the  predicted $\Omega^{(4)}(v)$ for further use. It reads
 \begin{equation}
 \label{eq:Omlin}
 \Omega^{(4)}(v)=b^{(4)}_0\Omega_0^{(4)}(v)+b^{(4)}_1\Omega_1^{(4)}(v)
 \end{equation}
 with
 \begin{equation}
 \label{eq:Om0}
 \Omega_0^{(4)}(v)=\frac{\Gamma(3/4)}{\Gamma(\zeta/4)}
 \sum_{l=0}^{\infty} \frac{\Gamma(l/2+\zeta/4)}{l!\, \Gamma(l/2+3/4)} \left( \frac{-v^2}{2\,\alpha^{1/2}}\right)^l
 \end{equation}
 and
\begin{equation}
 \label{eq:Om1}
\Omega_1^{(4)}(v)=\frac{v\sqrt{\pi/2}}{\Gamma[(\zeta+1)/4]} \sum_{l=0}^{\infty} \frac{\Gamma[(l+1+\zeta)/4]\, s(l)}{\Gamma(l/4+1)\,\Gamma[(l+3)/2]}
 \left[\frac{-v}{(4\alpha)^{1/4}}\right]^{l}\,,
 \end{equation}
where $s(l)$ is defined by
 \begin{equation}
 s(l)=2^{-1/2}\,\left[ \cos(l\pi/4) + \sin(l\pi/4)
 \right] \cos(l\pi/4)\,.
 \end{equation}
 
Aside from an overall (nonuniversal) amplitude $b^{(4)}_0$ and the nonuniversal scale $\alpha$, this scaling function involves a single universal parameter $\wp\equiv\alpha^{1/4}b^{(4)}_1/b^{(4)}_0$. To adjust it by means of their Monte Carlo results for the three-dimensional ANNNI model, Pleimling and Henkel \cite{PH01} considered ratios of truncated moment integrals $\tilde{M}_j(v_0) =\int_{v_0}^\infty d{v} \linebreak[0] \,v^j\, \Omega^{(4)}(v\alpha^{1/4})$, where the use of a lower integration limit $v_0>0$ was necessary because they were unable to compute numerically the function $\Omega(v)$ for values $v_0\lesssim 0.22$.
 
 In the next section, we re-examine Henkel's arguments leading to the scaling function~\eqref{eq:Omlin}. We will show that there are important reasons  to question the presence of a contribution proportional to $\Omega^{(4)}_1(v)$ in  $\Omega^{(4)}(v)$. 
 
 \section{Re-examination of the scaling-function solutions of the postulated invariance equations} \label{sec:rv}
 
 Let us return to the postulated invariance equations \eqref{eq:X0G}--\eqref{eq:X1G}. Unfortunately, it is not stated explicitly in Refs.~\cite{Hen99,Hen97,Hen02} in what region of $(t,\bm{r})$-space these are presumed to hold. Clearly, in the case of a bulk equilibrium systems with a LP, the obvious point of view would be to interpret them as being valid in full $d$-dimensional space $\mathbb{R}\times\mathbb{R}^{d-1}$, so that the $t$-variable is not restricted to positive values.%
\footnote{Obviously, this would be different for time-dependent phenomena such as relaxational processes where one must carefully distinguish between future and past time directions.}
Accepting this interpretation, we can solve these equations by Fourier transformation.%
\footnote{Needless to say that the position-space functions $G(t,r)$ can be trusted to belong to the space of tempered distributions (the dual of the Schwartz space $\mathcal{S}(\mathbb{R}^d)$ of rapidly decreasing $C^\infty$ functions), and hence to have well-defined Fourier transforms.} %
Equations~\eqref{eq:YG} and \eqref{eq:X0G} yield
\begin{equation}
   \label{eq:SymP1}
\left[ \bm{p}\, \partial_k +2\,\frac{\alpha}{N}\, k^{N-1} i^{N-2}\,\partial_{\bm{p}} \right] \tilde{G}(k,\bm{p})=0,
\end{equation}
and
\begin{equation}
   \label{eq:SymP2}
\left[ k\, \partial_k+\frac{N}{2}\,\bm{p} \cdot\partial_{\bm{p}} +N\,\tilde{\Delta}\right]\tilde{G}(k,\bm{p})=0,
\end{equation}
respectively, where the Fourier transform $\tilde{G}(k,\bm{p})$ is defined by
\begin{equation}
 \label{eq:FT}
G(t,r)=\int \frac{d k}{2\pi}\int \frac{d^{d-1}\bm{p}}{(2\pi)^{d-1}}\, e^{i (k t+\bm{p}\cdot\bm{r})}\tilde{G}(k,\bm p),
\end{equation}
and $\tilde{\Delta}$ means the scaling exponent
 \begin{equation}
 \label{eq:Deltatilde}
 \tilde{\Delta}=\frac{1}{N}-\Delta+\frac{d-1}{2}\,.
\end{equation}

The unique solution to the first-order partial differential equations (\ref{eq:SymP1}) and (\ref{eq:SymP2}) can be easily found by the method of characteristics. It reads
\begin{equation}
   \label{eq:gensolP}
 \tilde{G}(k,\bm{p})=\big[p^2+4 \alpha N^{-2} (ik)^N\big]^{-\tilde{\Delta}}.
 \end{equation}
Note that in the special case of  $N=4$ and  $\tilde{\Delta}=1$, the result reduces to the usual form $(p^2+\sigma k^4)^{-1}$ of the free momentum-space propagator at a LP (with $\sigma= \alpha/4$)   [see e.g.\ Refs. \cite{DS00a,SD01,Die02} and Eq.~\eqref{eq:Gf} below].

Let us assume that $4\alpha (ik)^N/N^2>0 $. For even $N\geq 4$, this is the case when $\alpha\,(-1)^{N/2} >0$. The Fourier backtransform of the function~\eqref{eq:gensolP} with respect to the variable $\bm{p}$ may be gleaned from Ref.~\cite[p.\ 288]{GS64}. Using this, one finds that the scaling function is given by
 \begin{equation}
   \label{eq:frX}
 \Omega^{(N)}(v)= C_N(\zeta,\alpha) \int_0^\infty d\kappa\, \kappa^{(\zeta-1)/2}
 \cos\bigg[\bigg(\frac{N^2}{4|\alpha|}\bigg)^{1/N} v\, \kappa\bigg] \, K_{(\zeta-1)/N}(\kappa^{N/2}) 
 \end{equation}
 with
 \begin{equation}
 C_N(\zeta,\alpha)=\frac{\pi ^{-(d+1)/2} 2^{2-d+(\zeta-1)/N}}{\Gamma \left[(d-1)/2-(\zeta-1)/N\right]}\,\,\bigg(\frac{N^2}{4|\alpha|}\bigg)^{1/N},
\end{equation}
 where $K_{\nu}(z)$ is the Macdonald function (modified Bessel function of the second kind).
The function $ \Omega^{(N)}(v)$ is a (particular) solution of the ordinary differential equation (\ref{eq:Hv}).  By construction, it is the unique scaling function (up to scales) consistent with the validity of Eqs.~\eqref{eq:X0G}, \eqref{eq:YG}, and \eqref{eq:Gscfs} in full $(t,\bm{r})$-space $\mathbb{R}^d$.

Let us consider the case of our primary concern,  $N=4$ with $\alpha>0$, in more detail. The differential equation~(\ref{eq:Hv}) then is of third order. Using {\sc Mathematica} \cite{Mathematica7}, one easily arrives at the three linearly independent solutions
 \begin{eqnarray}
 \label{eq:fundsyst}
 f_1(v)&=&\, _1F_2{\left(\frac{\zeta }{4};\frac{1}{2},\frac{3}{4};\frac{v^4}{16\alpha}\right)},\nonumber\\
f_2(v)&=& \, _1F_2{\left(\frac{\zeta}{4}+\frac{1}{4};\frac{3}{4},\frac{5}{4};\frac{v^4}{16 \alpha}\right)}\,v,\nonumber\\
f_3(v)&=& \, _1F_2{\left(\frac{\zeta}{4}+\frac{1}{2};\frac{5}{4},\frac{3}{2};\frac{v^4}{16 \alpha}\right)}\,v^2.
\end{eqnarray}
Moreover, the power series of $\Omega_0^{(4)} (v)$ and $\Omega_1^{(4)}(v)$ given by Eqs.~\eqref{eq:Om0}  and \eqref{eq:Om1} with $N=4$ can be summed explicitly and the integral~\eqref{eq:frX} for $\Omega^{(4)}(v)$ be computed to obtain the results
\begin{eqnarray}
 \label{eq:O40}
\Omega^{(4)}_0(v)&=&  f_1(v)-\frac{2\,\Gamma(3/4)\,  \Gamma[(\zeta +2)/4]}{\sqrt{\alpha}\,\Gamma(1/4)\, \Gamma(\zeta /4)}\,f_3(v),\nonumber\\
\Omega^{(4)}_1(v)&=& f_2(v)-\frac{(2 \pi )^{1/2} \, \Gamma[(\zeta
   +2)/4]}{\alpha^{1/4} \,\Gamma(1/4)\,
   \Gamma[(\zeta +1)/4]}\,f_3(v),
    \end{eqnarray}
and
 \begin{equation}
 \label{eq:Om4int}
\Omega^{(4)}(v)=C_4(\zeta,\alpha)\,2^{(\zeta-3)/4}\,\Gamma(5/4)\,\Gamma(\zeta/4)\,\Omega^{(4)}_0(v).
\end{equation}

Note that the integrals~\eqref{eq:frX} are even in $v$. Hence  $\Omega^{(4)}$ cannot have a contribution proportional to  the odd solution $f_2$.  Since $\Omega^{(4)}_1$ has a term $\propto f_2$, it also cannot contribute to $\Omega^{(4)}$. That is, if the invariance equations~\eqref{eq:X0G}--\eqref{eq:X1G} are taken to hold in  full $(t,\bm{r})$-space $\mathbb{R}^d$,  the coefficient $b^{(4)}_1$ in Eq.~\eqref{eq:Omlin} must vanish, as evidenced by our explicit result~\eqref{eq:Om4int}.
 
There is a further reason by which a contribution $\propto \Omega^{(4)}_1(v)$ to $\Omega^{(4)}(v)$ is ruled out: The functions  $\Omega^{(4)}_0(v)$ and  $\Omega^{(4)}_1(v)$ both  satisfy the boundary condition~\eqref{eq:Bcinf} and hence vary as $v^{-\zeta}$ as $v\to+\infty$. Being even in $v$, the former also vanishes $\sim |v|^{-\zeta}$ as $v\to-\infty$. By contrast, $\Omega^{(4)}_1(v)$ \emph{diverges exponentially} in this limit. A convenient way to obtain its asymptotic behavior is to use its integral representation
 \begin{equation}
\label{eq:Om1intrep}
\Omega^{(4)}_1(v)=\frac{-(4/\alpha)^{1/4}\,2^{(7-\zeta)/4}}{(2\pi)^{1/2}\,\Gamma[(1+\zeta)/4]}\,P\big[(4/\alpha)^{1/4}v\big]
\end{equation}
with
 \begin{equation}
 \label{eq:Jdef}
  P(v)= \int_0^\infty d\kappa\, \kappa^{(\zeta-1)/2} [e^{- v \kappa}-\cos (  v \kappa)-\sin( v \kappa)] 
  \, K_{(\zeta-1)/4}(\kappa^2).  
 \end{equation}
 In the limit $v\to-\infty$, we can replace the Bessel function $K_{(\zeta-1)/4}(\kappa^2)$ by its limiting form $\sqrt{\pi/(2\kappa^2)}\,e^{-\kappa^2}$, ignore the contributions from the parts of the integrand proportional to $\cos(|v|\kappa)$ and $\sin(|v|\kappa)$, and determine the dominant contribution of the remaining integral by expanding the argument of the exponential about the saddle point $\kappa=|v|/2$. This gives
  \begin{eqnarray}
 \label{eq:Om1mininf}
P(v)&\mathop{\approx}_{v\to-\infty}& \int_0^\infty d\kappa\,\kappa^{(\zeta-1)/2}\,\sqrt{\frac{\pi}{2\kappa^2}}e^{\kappa|v|-2\kappa^2} \nonumber\\
&\mathop{\approx}_{v\to-\infty}&\sqrt{\pi/2}\int_{-\infty}^\infty dx\,(|v|/2)^{(\zeta-3)/2}\,e^{v^2/4-x^2}=
2^{1-\zeta/2}\pi\,|v|^{(\zeta-3)/2}\,e^{v^2/4}.
\end{eqnarray}

It may be tempting to argue that the scaling form~\eqref{eq:Gscfs} should be modified by replacing $t$ by its absolute value in the scaling function $\Omega^{(4)}$, writing
\begin{equation}
\label{eq:modGscf}
G(t,r)=r^{-2\Delta}\,\Omega^{(4)}\big(|t|r^{-1/2}\big).
\end{equation}
 In this way, the divergence of a contribution proportional to $\Omega^{(4)}_1(|t|r^{-1/2})$ for $t\to-\infty$ would be avoided. However, whenever the coefficient $b_1^{(4)}$ of $\Omega^{(4)}_1$ does not vanish, such a correlation function fails to satisfy equation (\ref{eq:YG}) in full $(t,\bm{r})$-space $\mathbb{R}^d$. Indeed, application of the operator  $\bm{Y}_{-1}$ to $G(t,r)$ yields
 \begin{equation}
 \label{eq:y-1}
 \bm{Y}_{-1}\,G(t,r)=-\frac{\alpha\, b^{(4)}_1}{r^{(4\Delta+1)/2}}\bm{r}\,\delta'(t) 
 \end{equation}
rather than zero. Note that for reasons discussed in the paragraph following Eq.~\eqref{eq:X1G}, this conclusion does \emph{not} hinge on the definition of fractional derivatives $\partial_t^{N-1}$ when $N$ is not a natural number. Using Henkel's definition of fractional derivatives given in Ref.~\cite[Appendix A]{Hen02}, one can in fact determine $\lim_{N\to4} \bm{Y}_{1-N/2}G(t,r)$ in a straightforward fashion to recover the same result, Eq.~\eqref{eq:y-1}. We are grateful to Malte Henkel (private correspondence) and one referee who both went explicitly through this analysis, confirming our conclusion that the function~\eqref{eq:modGscf} satisfies the inhomogeneous equation~\eqref{eq:y-1} rather than its homogeneous counterpart, Eq.~\eqref{eq:YG} with $N=4$. 

The proposed result~\eqref{eq:Omlin} for the scaling function $\Omega^{(4)}$, when re-interpreted according to Eq.~\eqref{eq:modGscf}, does therefore not satisfy the original equation given in Henkel's work \cite{Hen97,Hen02}. Let us nevertheless accept the predictions for $G(t,r)$ specified by Eqs.~(\ref{eq:Omlin})--(\ref{eq:Om1}) and (\ref{eq:modGscf}) for the moment%
\footnote{One motivation for this acceptance is the previously mentioned remarkably good agreement of the corresponding predictions for the  order-parameter two-point function with the results of Monte Carlo simulations for the three-dimensional ANNNI model \cite{PH01}. }
and work out their consequences  for the Fourier transform $\tilde{G}(k,\bm{p})$. The results will be checked against explicit RG results in $4-\epsilon$ dimensions and shown to be incompatible with the latter in subsequent sections.

Rather than starting directly from these equations, it is more convenient to go back to the partial differential equations  (\ref{eq:X0G}) and (\ref{eq:y-1}). Their Fourier transforms are Eqs.~(\ref{eq:SymP2}) (with $N$ set to $4$),  and
\begin{eqnarray} 
\label{eq:PaLe}
  \big( \bm{p}\, \partial_k -\frac{  \alpha}{2}\, k^{3} \,
  \partial_{\bm{p}} \big) \tilde{G}(k,\bm{p})=  
   \alpha\, b_1^{(4)}\,A\,(1-2\tilde{\Delta})\,k\,\frac{ \bm{p}}{p}\, p^{-2 \tilde{\Delta}}.
 \end{eqnarray}
 Here
 \begin{equation}
  A=\int d^{d-1}r\,r^{-(4\Delta+1)/2}\, e^{-i \bm{e}\cdot\bm{r}},
 \end{equation}
 where $\bm{e}\in  \mathbb{R}^{d-1}$ is a unit vector. Equation~\eqref{eq:SymP2} yields the scaling form\begin{equation}
\label{eq:scP}
  \tilde{G}(k,\bm{p})=p^{-2\tilde{\Delta}}\,g\big(p^2k^{-4}\big).
\end{equation}
Substituting this into Eq.~\eqref{eq:PaLe} gives us a differential equation for the scaling function $g(w)$, namely
\begin{equation}
\label{eq:odP}
 w^{1/2}\,[-(\alpha+4w)\,
 g'(w)+\alpha\,\tilde{\Delta}\, w^{-1}\,g(w)]=
\alpha \,b_1^{(4)} A\,(1-2\, \tilde{\Delta}).
 \end{equation}
This is solved in a straightforward fashion to obtain
\begin{eqnarray}
\label{eq:Solg}
 g(w)&=&{\bigg(1+\frac{\alpha}{4w}\bigg)}^{-\tilde{\Delta}}\,\Bigg[a +
 \frac{\alpha}{4}\,b_1^{(4)} A\,(1-2\, \tilde{\Delta}) \, \int_w^\infty \frac{d
 u}{u^{3/2}}\,{\bigg(1+\frac{\alpha}{4u}\bigg)}^{\tilde{\Delta}-1}\Bigg]\nonumber\\
 &=&{\bigg(1+\frac{\alpha}{4w}\bigg)}^{-\tilde{\Delta}}\,\left[a +
 \frac{\alpha}{2}\,b_1^{(4)} A\,(1-2\, \tilde{\Delta}) \,w^{-1/2}\,_2F_1{\left(\frac{1}{2},1-\tilde{\Delta};\frac{3}{2};-\frac{\alpha}{4w}\right)}\right].\nonumber\\
 \end{eqnarray}
Here  $a$ is a constant which must  depend linearly on the coefficients $b^{(4)}_0$ and
 $b^{(4)}_1$ of Eq.~\eqref{eq:Omlin}: $a=a_0 b^{(4)}_0+a_1
 b^{(4)}_1$. The explicit expressions for the coefficients $a_0$ and
 $ a_1$ are not important for us. 
 
 It is now easy to see that the contribution $\propto b_1^{(4)}$ is unacceptable. The function $g(w)$ has the asymptotic behavior
 \begin{equation}
 g(w)\underset{w\to 0}{\approx} \text{const}\,w^{\tilde{\Delta}}-2A\,b_1^{(4)}\,\sqrt{w}.
 \end{equation}
Upon substituting this into Eq.~\eqref{eq:scP}, we arrive at the limiting form
 \begin{equation}\label{eq:Gkpgrt0exp}
 \tilde{G}(k,p)\underset{p\to 0}{\approx}\text{const}\,k^{-4\tilde{\Delta}}-2A\,b_1^{(4)}\,p^{1-2\tilde{\Delta}}\,k^{-2}.
 \end{equation}
But for nonvanishing momentum $k$ the function $G(k,p)$ must have a Taylor expansion in $p$ of the form
 \begin{equation}\label{eq:pexp}
 \tilde{G}(k>0,p)=G_0(k)+G_2(k)\,p^2+\Or(p^4),
 \end{equation}
 where the behaviors $G_0(k)\sim k^{-2\tilde{\Delta}/\theta}$ and $G_2(k)\sim k^{-(2+2\tilde{\Delta})/\theta}$ are dictated by scaling. To understand the behavior of $G_2(k)$, note that the scaling dimension of the descendant operator $\partial_{\bm{r}}\bm{\phi}(t,\bm{r})$ is $\Delta+1$. This translates into the stated behavior of $G_2(k)$ upon Fourier transformation. Furthermore, a $k$-dependent second-moment correlation length $\xi_\beta(k)$ governing the decay of $G(t,r)$ as $r\to\infty$ can be defined via its square
 \begin{equation}
 \xi^2_\beta(k)=-\partial_{p^2}\ln \tilde{G}(k,p)\big|_{p=0}=\frac{1}{2}\,
\frac{\int d^mt\,e^{-i\bm {k}\cdot\bm{ t} } \int d^{d-m}r\,r^2\,G(t,r)}{\int d^mt\,e^{-i\bm{k}\cdot\bm{t} } \int d^{d-m}r\,G(t,r)},
 \end{equation}
 where $m=1$ in the uniaxial case we are concerned with. The length
$\xi_\beta(k)$ scales as $k^{-1/\theta}$; it has a finite value as long as $k$ does not vanish. Thus, $G_{\bm{k}}(r)\equiv\int d^mt \,e^{-i\bm{k}\cdot\bm{t} }\,\tilde{G}(t,r)$ must decay exponentially on the scale of $\xi_\beta(k)$ in the large-$r$ limit, 
and all even moments $\int d^{d-m}r\,r^{2j} \,G_{\bm{k}}(r) $ with $j=1,2,\dotsc$ should exist. 
This in turn means that $\tilde{G}(k>0,p)$ must be analytic in $p^2$ at $p=0$. Analogous considerations show that $\tilde{G}(k,p>0)$ must be analytic in $k$ at $k=0$ and expandable as 
  \begin{equation}\label{eq:kexp}
 \tilde{G}(k,p>0)=g_0(p)+g_2(p)\,k^2+\Or(k^4),
 \end{equation}
 where $g_0(p)\sim p^{-2\tilde{\Delta}}$ and  $g_2(p)\sim p^{-2\tilde{\Delta}-2\theta}$ as $p\to 0$.
 
 The contribution $\propto b_1^{(4)}$ in Eq.~\eqref{eq:Gkpgrt0exp} yields a term $\sim p^{1-2\tilde{\Delta}}k^{-2}$ that is nonanalytic in $p$ at $p=0$. It is inconsistent with the expansion~\eqref{eq:pexp} and violates the mentioned analyticity in $p$. Consequently, this contribution must not appear, and the coefficient $b_1^{(4)}$ should be zero.
  
 \section{Model, renormalization, and renormalization-group equations} \label{sec:model}
\subsection{$O(n)$ model for critical behavior at $m$-axial Lifshitz points}
According to Ref.~\cite{PH01}, the Monte Carlo results obtained for the three-dimensional ANNNI model at its uniaxial LP could be fitted very well to the scaling functions proposed by Henkel for the case with $N=4$, i.e.\ $\theta=1/2$. From the $\epsilon$-expansion results of Refs.~\cite{DS00a,SD01} it is well known that the anisotropy exponents $\theta(d,n,m)$ of $m$-axial LP differ from $1/2$  at order $\epsilon^2$.  One therefore expects that $\theta\ne 1/2$ also for the three-dimensional ANNNI model. However, the difference $\theta-1/2$ appears to be small. Field-theory estimates based on the $\epsilon$ expansion gave $\theta(3,1,1)\simeq 0.47$ \cite{SD01,Die05}, and Monte Carlo simulation data seem to be consistent with values $0.48\lesssim\theta\lesssim 1/2$ \cite{rem:thetagen,PH01}. 
 
Our aim here is to determine appropriate two-point correlation functions by means of RG improved perturbation theory in  $d<d^*(m)$ dimensions and compare the results with Henkel's predictions. To  this end, we will consider  the model used in the two-loop RG analysis of Refs.~\cite{DS00a,SD01}. Its Hamiltonian is given by
 \begin{equation}
   \label{eq:Hamf}
     {\mathcal{H}}=\int d^dx\bigg[\frac{\mathring{\sigma}}{2}\,(\partial_\alpha^2\bm{\phi})^2+\frac{1}{2}(\partial_\beta\bm{\phi})^2+\frac{\mathring{\rho}}{2}\,(\partial_\alpha\bm{\phi})^2
       +\frac{\mathring{\tau}}{2}\,
 \bm{\phi}^2+\frac{\mathring{u}}{4!}\,|\bm{\phi} |^4
     \bigg],
 \end{equation}
 where $\bm{\phi}(\bm{x})=(\phi_a(\bm{x}),a=1,\ldots,n)$ is an $n$-component order-parameter field. Pairs of $\alpha$ and $\beta$ indices are to be summed from $1$ to $m$ and from $m+1$ to $d$, respectively. Thus, $\partial_\alpha^2$ denotes the Laplacian $\partial_{\bm{z}}^2=\partial_{\bm{z}}\cdot\partial_{\bm{z}}$ in $\bm{z}$-space, while $(\partial_\beta\bm{\phi})^2=(\partial_\beta\phi_a)\partial_\beta\phi_a$ is the square of the gradient $\partial_{\bm{r}}\bm{\phi}$ in $\bm{r}$-space. We shall investigate the two-point cumulants of the order parameter $\bm{\phi}$ and the energy density $\bm{\phi}^2(\bm{x})/2$. They are defined by
 \begin{equation}
\delta_{a b}\,G^{(2,0)}({\bm x}-{\bm x'})= \langle\phi_a ({\bm x})\,\phi_b ({\bm x'})\rangle^{\mathrm{cum}}
 \end{equation}
 and
 \begin{equation}
G^{(0,2)}({\bm x}-{\bm x'})\equiv\bigg\langle\frac{1}{2}\bm{\phi}^2(\bm{x})\,\frac{1}{2}\bm{\phi}^2 ({\bm x'})\bigg\rangle^{\mathrm{cum}}, 
 \end{equation}
respectively.

Despite the fact that $\theta\ne1/2$, we shall be able to perform nontrivial checks of Henkel's predictions. The reason is the following. Building on the work in Refs.~\cite{DS00a,SD01}, we can use the $\epsilon$ expansion about the upper critical dimension $d^*(m)$ to investigate the behaviors of $G^{(0,2)}$ and $G^{(2,0)}$ at the LP. The critical exponents these functions involve (such as $\theta$ and $\Delta$) as well as their scaling functions have expansions in powers of $\epsilon$. Since the anisotropy exponent $\theta$ starts to deviate from $1/2$ not before second order in   $\epsilon$, Henkel's predictions --- if correct --- ought to comply with the $\Or(\epsilon)$ results of the $\epsilon$ expansion. In the case of the order-parameter cumulant $G^{(2,0)}$, the contribution of zeroth order in $\epsilon$ of the scaling function is given by Landau theory; the leading corrections are of order $\epsilon^2$ and encountered at two-loop order. By contrast, the one-loop term of the energy-density cumulant $G^{(0,2)}$ yields a contribution of zeroth order in $\epsilon$ to its scaling function.%
\footnote{An analogous situation is encountered in checks of the LSI predictions for type-II systems \cite{CG02b}. However, in this case the $\Or(\epsilon^2)$ term of the order-parameter response function could be determined explicitly and shown to be in conflict with the LSI prediction. Energy-density response and correlation functions were not computed in this work. By analogy with our case, the scaling functions of these dynamic cumulants will receive nontrivial corrections already at first order in $\epsilon$. The latter may be expected to show deviations from the LSI predictions as well.}
As we shall see, the latter is inconsistent with Henkel's predictions.

\subsection{Renormalization}
To proceed, it will be necessary to recall some background on the RG analysis of the above model with Hamiltonian~\eqref{eq:Hamf}. 
The renormalization of the (dimensionally regularized) $M$-point cumulants (connected correlation functions)
\begin{equation}
 \label{sec:GMdef}
 G_{a_1,\ldots,a_M}^{(M)}(\bm{x}_1,\ldots,\bm{x}_M)=\langle\phi_{a_1}(\bm{x}_1)\cdots\phi_{a_M}(\bm{x}_M)\rangle^{\mathrm{cum}}
\end{equation}
of this theory in $d=d^*(m)-\epsilon\le d^*(m)$ dimensions has been explained in detail in Refs.~\cite{DS00a,SD01,DGR03}. Their UV divergences can be absorbed by means of the reparametrizations
\begin{eqnarray}\label{eq:bulkrep}
  \bm{\phi}=Z_\phi^{1/2}\,\bm{\phi}_{\mathrm{R}}\;,\quad
  \mathring{\sigma}=Z_\sigma\,\sigma\;,
\quad  \mathring{u}\,{\mathring{\sigma}}^{-m/4}\,F_{m,\epsilon}=
   \mu^\epsilon\,Z_u\,u\;,\nonumber\\
\mathring{\tau}-\mathring{\tau}_{\mathrm{LP}}=
\mu^2\,Z_\tau{\big(\tau+A_\tau\,\rho^2\big)}\;,
\quad
\left(\mathring{\rho}-\mathring{\rho}_{\mathrm{LP}}\right)\,
{\mathring{\sigma}}^{-1/2}=\mu\,Z_\rho\,\rho\,,
\end{eqnarray}
where $\mu$ is a momentum scale. Following these references, we choose the factor  $F_{m,\epsilon }$  as
\begin{equation}
  \label{eq:Fmeps}
F_{m,\epsilon}=
\frac{\Gamma{\left(1+{\epsilon/ 2}\right)}
\,\Gamma^2{\left(1-{\epsilon/ 2}\right)}\,
\Gamma{\left({m}/{4}\right)}}{(4\,\pi)^{({8+m-2\,\epsilon})/{4}}\,
\Gamma(2-\epsilon)\,\Gamma{\left({m}/{2}\right)}}\,.
\end{equation}
The LP is located at $(\mathring{\tau},\mathring{\rho})=(\mathring{\tau}_{\mathrm{LP}},\mathring{\rho}_{\mathrm{LP}})$. In a theory regularized by means of a large-momentum cutoff $\Lambda$, the renormalization functions $\mathring{\tau}_{\mathrm{LP}}$ and $\mathring{\rho}_{\mathrm{LP}}$  would diverge $\sim\Lambda^2$ and $\sim \Lambda$, respectively. However, in our perturbative approach based on dimensional regularization, they vanish. Results to order $u^2$ for the renormalization factors $Z_\phi$, $Z_\sigma$, $Z_\rho$, $Z_\tau$ and $Z_u$ can be found in Eqs.~(40)--(50) of Ref.~\cite{SD01}. The function $A_\tau$ is given to $\Or(u)$ in Eq.~(17) of Ref.~\cite{DGR03}. Since we will not move away from the LP, we shall not need it.

Using the reparametrizations~\eqref{eq:bulkrep}, we can define renormalized $M$-point cumulants by
\begin{equation}
 \label{eq:GMren}
G^{(M)}_{a_1,\ldots,a_M;\mathrm{R}}(\bm{x}_1,\ldots;\tau,\rho,\sigma,u,\mu)=Z_\phi^{-M/2}\,G^{(M)}_{a_1,\ldots,a_M}(\bm{x}_1,\ldots;\mathring{\tau},\mathring{\rho},\mathring{\sigma},\mathring{u}).
\end{equation}
When $M=2$, this defines us the renormalized function $G^{(2,0)}_{\mathrm{R}}\equiv G^{(2)}_{\mathrm{R}}$. However, the renormalization of the energy-density cumulant $G^{(0,2)}$ also involves an additive counterterm. A possible way of fixing it is to subtract from $\tilde{G}^{(0,2)}$ its value at a  normalization point ($\mathrm{NP}$). We choose $\mathrm{NP}$ at the LP and a momentum $(\bm{k},\bm{p})=(\bm{0},\mu\,\hat{\bm{p}})$, where $\hat{\bm{p}}$ is an arbitrary $(d-m)$-dimensional unit vector, defining  
\begin{equation}
 \label{eq:G02ren}
\tilde{G}^{(0,2)}_{\mathrm{R}}(\bm{k},\bm{p};\tau,\rho,\sigma,u,\mu)= Z^2_\tau\big[\tilde{G}^{(0,2)}(\bm{k},\bm{p};\mathring{\tau},\mathring{\rho},\mathring{\sigma},\mathring{u})-B(\mathring{u},\mathring{\sigma};\mu)\big]
\end{equation}
with
\begin{equation}
 \label{eq:Csubdef}
B(\mathring{u},\mathring{\sigma};\mu)= \tilde{G}^{(0,2)}\big|_{\mathrm{NP}}\equiv \tilde{G}^{(0,2)}\big|_{k=0,p=\mu;\mathrm{LP}}.
\end{equation}
Hence the renormalized function $\tilde{G}^{(0,2)}_{\mathrm{R}}$ satisfies the normalization condition
\begin{equation}
 \label{eq:ednormcond}
\tilde{G}^{(0,2)}_{\mathrm{R}}\big|_{\mathrm{NP}}\equiv\tilde{G}^{(0,2)}_{\mathrm{R}}(\bm{0},\mu\,\hat{\bm{p}};\tau=0,\rho=0,\sigma,u,\mu)=0.
\end{equation}

\subsection{Renormalization-group equations}

The RG equations one obtains from Eqs.~\eqref{eq:bulkrep} and \eqref{eq:GMren} upon varying $\mu$ are known from Refs.~\cite{DS00a} and \cite{SD01}. Let us introduce the operator
\begin{equation}
\label{eq:Dmu}
\mathcal{D}_\mu=\mu\partial_\mu+\sum_{g=u,\sigma,\rho,\tau}\beta_g\partial_g
\end{equation}
along with the beta functions
\begin{equation}
\label{eq:betagdef}
\beta_g\equiv\mu\partial_\mu|_0g\,,\;\;g=u,\sigma,\rho,\tau,
\end{equation}
where $\partial_\mu|_0$ means a derivative at fixed bare interaction constant $\mathring{u}$ and parameters  $\mathring{\sigma}$, $\mathring{\rho}$, and $\mathring{\tau}$. The $\beta_g$ can be expressed in terms of the exponent functions
\begin{equation}
\label{eq:etas}
\eta_g(u)=\mu\partial_\mu |_0\ln{Z}_g\,,\;g=\phi,u,\sigma,\rho,\tau,
\end{equation}
and the function
\begin{equation}
\label{eq:btau}
b_\tau(u)=A_\tau[\mu\partial_\mu |_0\ln A_\tau+\eta_\tau-2\eta_\rho]
\end{equation}
as
\begin{equation}
\label{eq:betags}
\beta_g=
\begin{cases}
-u[\epsilon+\eta_u(u)],&g=u,\\
-\sigma\eta_\sigma(u),&g=\sigma,\\
-\rho[1+\eta_\rho(u)],&g=\rho,\\
-\tau[2+\eta_\tau(u)]-b_\tau(u)\rho^2,&g=\tau.
\end{cases}
\end{equation}
Two-loop results for the functions $\beta_u$, $\eta_\phi$, $\eta_\sigma$, $\eta_\rho$, and $\eta_\tau$ may be gleaned from  Eqs.~(59), (58), and (40)--(50) of Ref.~\cite{SD01}. The function $b_\tau(u)$ (not needed in the following) is given to first order in $u$  in Eq.~(40) of Ref.~\cite{DR04}.

The RG equations for the functions $G^{(M)}_{\mathrm{R}}$ can be written as
\begin{equation}
\label{eq:RGGM}
\left(\mathcal{D}_\mu+\frac{M}{2} \,\eta_\phi\right)G^{(M)}_{a_1,\ldots,a_N;\mathrm{R}}(\bm{x}_1,\ldots,\bm{x}_N)=0.
\end{equation}
Owing to the additive counterterm, the RG equation for $\tilde{G}^{(0,2)}_{\mathrm{R}}$ becomes inhomogeneous; it reads
\begin{equation}
\label{eq:RGGed}
\left(\mathcal{D}_\mu-2 \,\eta_\tau\right)\tilde{G}^{(0,2)}_{\mathrm{R}}(\bm{k},\bm{p})=\mu^{-\epsilon}\sigma^{-m/4}\mathcal{B}(u),
\end{equation}
where $\mathcal{B}(u)$ is a UV finite function defined by
\begin{equation}
\label{eq:Bcaldef}
\mu^{-\epsilon}\sigma^{-m/4}\mathcal{B}(u)=-Z_\tau^2\,\mu\partial_\mu \big|_0 B(\mathring{u},\mathring{\sigma};\mu).
\end{equation}

The RG Eqs.~\eqref{eq:RGGM} and \eqref{eq:RGGed} can be solved in a standard fashion using characteristics \cite{DS00a,SD01,DGR03,DR04}. For our purposes it will be sufficient to focus on the solutions to the RG equations of $\tilde{G}^{(2,0)}_{\mathrm{R}}$ and $\tilde{G}^{(0,2)}_{\mathrm{R}}$ at the LP. The critical exponents they involve --- namely, the correlation exponent $\eta_{\mathrm{L}2}$, the correlation-length exponent $\nu_{\mathrm{L}2}$, and the anisotropy exponent $\theta$ --- may be expressed in terms of the values $\eta_g^*\equiv\eta_g(u^*)$ of  the exponent functions $\eta_g(u)$ at the infrared-stable zero 
\begin{equation}
\label{eq:ustar}
u^*=\frac{2\epsilon}{3}\,\frac{9}{n+8}+u_2^*\,\epsilon^2+\Or(\epsilon^3)
\end{equation}
 of $\beta_u$ whose coefficient $u^*_2$, albeit not needed here, may be found in Eq.~(60) of Ref.~\cite{SD01}. We have
\begin{equation}
 \label{eq:etaL2nuL2thetadef}
\eta_{\mathrm{L}2}=\eta_\phi^*\,,\quad\nu_{\mathrm{L}2}=(2+\eta_\tau^*)^{-1}\,,\mbox{ and }\theta=(2+\eta_\sigma^*)/4\,.
\end{equation}
 Solving the RG equation for  $\tilde{G}^{(2,0)}_{\mathrm{R}} $   at the LP $\tau=\rho=0$ gives  
 \begin{equation}
 \label{eq:RGsolG20}
\tilde{G}^{(2,0)}_{\mathrm{R}}(\bm{k},\bm{p})\approx\mu^{-2}(E_\phi^*)^{-2}(p/\mu)^{\eta_{\mathrm{L}2}-2}\,\Psi_{m,d}\big[(E_\sigma^*\sigma)^{1/4}\mu^{-1/2} k\,(p/\mu)^{-\theta}\big]
\end{equation}
with
\begin{equation}
 \label{eq:scfPsi}
\Psi_{m,d}(\mathsf{k})=\tilde{G}^{(2,0)}_{\mathrm{R}}(\bm{\mathsf{k}},\hat{\bm{p}};\tau=0,\rho=0,\sigma=1,u^*,\mu=1).
\end{equation}
Here $E_\sigma^*=E_\sigma^*(u)$ is a familiar nonuniversal amplitude. Just as its analog $E_\tau^*(u)$, which we will encounter when solving the RG equation for $G^{(0,2)}_{\mathrm{R}}$, it can be expressed as an  integral along a RG trajectory:
\begin{equation}
E_g^*(u)\equiv E_g(u^*,u),\quad E_g(\bar{u},u)=\exp\left[\int_u^{\bar{u}}du'\frac{\eta_g^*-\eta_g(u')}{\beta_u(u')}\right],\;\;g=\sigma,\tau,\phi.
\end{equation}

The solution of the inhomogeneous RG equation~\eqref{eq:RGGed} for the energy-density correlation function gives
\begin{equation}
 \label{eq:RGsolG02}
\tilde{G}^{(0,2)}_{\mathrm{R}}(\bm{k},\bm{p})\approx\frac{(E_\tau^*)^2\,\mathcal{B}(u^*)\,\nu_{\mathrm{L}2}}{\mu^\epsilon\,(E_\sigma^*\sigma)^{m/4}\,\alpha_{\mathrm{L}}}\bigg\{\left(\frac{p}{\mu}\right)^{-\alpha_L/\nu_{\mathrm{L}2}}\Upsilon_{m,d}\big[(E_\sigma^*\sigma)^{1/4}\mu^{-1/2} k\,(p/\mu)^{-\theta}\big] - 1\bigg\}
\end{equation}
with
\begin{equation}
 \label{eq:scfUps}
\Upsilon_{m,d}(\mathsf{k})=1+\frac{\alpha_{\mathrm{L}}}{\mathcal{B}(u^*)\,\nu_{\mathrm{L}2}}\,\tilde{G}^{(0,2)}_{\mathrm{R}}(\bm{\mathsf{k}},\hat{\bm{p}};\tau=0,\rho=0,\sigma=1,u^*,\mu=1) 
\end{equation}
where 
 \begin{equation}
\alpha_{\mathrm{L}}=2-(d-m+m\, \theta)\nu_{\mathrm{L}2}
\end{equation}
is the specific heat exponent \cite{SD01}.
 
\section{Two-loop calculation of scaling functions}\label{sec:2lres}
\subsection{Calculations and results for general values of $m$}
\label{sec:genm}
We next turn to the calculation of the scaling functions~\eqref{eq:scfPsi} and \eqref{eq:scfUps} by means of RG improved perturbation theory.

The Fourier transform of the free propagator $G_f(\bm{x})$ reduces  at the LP to the simple expression
\begin{equation}
 \label{eq:Gf}
\tilde{G}_f(\bm{k},\bm{p})=\big[p^2+\mathring{\sigma}k^4\big]^{-1}.
\end{equation}
Unfortunately, its Fourier backtransform yields for general values of $m$ and $d$ a rather complicated expression for the scaling function of $G_f(\bm{x})$. One finds \cite{DS00a,SD01}
\begin{equation}
 \label{eq:Gfscf}
G_f(\bm{x})=r^{-2+\epsilon}\,\mathring{\sigma}^{-m/4}\,
    \Phi_{m,d}{\big(\mathring{\sigma}^{-1/4}zr^{-1/2}\big)}
\end{equation}
with
\begin{eqnarray}
   \label{eq:Phi}
   \Phi_{m,d}(\upsilon)&=&\frac{1}{2^{2+m}\,\pi^{(6+m-2\epsilon)/4}}\,
   {\bigg[}
   \frac{\Gamma(1- \epsilon/2)}{\Gamma[(m+2)/4]}\,
 {}_1{F}_{2}\Big(1-\frac{\epsilon}{2};\frac{1}{2},\frac{m+2}{4};
     \frac{\upsilon^4}{64}\Big)
     \nonumber\\ && \qquad\qquad\qquad\strut
 -\frac{\upsilon^2\,\Gamma[(3-\epsilon)/2]}{4\,
    \Gamma(1+m/4)}\,
 {}_1{F}_{2}\Big(\frac{3-\epsilon}{2};\frac{3}{2},1+\frac{m}{4};
         \frac{\upsilon^4}{64}\Big){\bigg]}.
 \end{eqnarray}

Owing to our use of dimensional regularization, the contribution from the one-loop graph \;\raisebox{-0.7ex}{\includegraphics[width=16pt,clip]{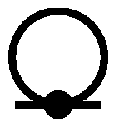}}\; vanishes at the LP. 
Hence the perturbation expansions to two-loop order of the cumulants $\tilde{G}^{(2,0)}(\bm{k}, \bm{p})$ and  $\tilde{G}^{(0,2)}(\bm{k}, \bm{p})$ at the LP become
 \begin{equation}
 \label{eq:Exod}
 \big[\tilde{G}^{(2,0)}(\bm{k}, \bm{p})\big]^{-1}=p^2+\mathring{\sigma}k^4 -
\frac{\mathring{u}^2}{6}\,\frac{n+2}{3}
 J_3(\bm{k}, \bm{p};\mathring{\sigma})+\Or(\mathring{u}^3)
 \end{equation}
 and
  \begin{equation}
 \label{eq:Exen}
 \tilde{G}^{(0,2)}(\bm{k},\bm{p})=\frac{n}{2}\,
 J_2(\bm{k}, \bm{p};\mathring{\sigma})-\frac{{\mathring{u}}\,n(n+2)}{12}\, [J_2(\bm{k}, \bm{p};\mathring{\sigma})]^2+\Or(\mathring{u}^2), 
\end{equation}
 respectively. Here $J_{j=2,3}$  are the integrals
\begin{eqnarray} 
  \label{eq:Jldefx}
 J_j(\bm{k},\bm{p};\mathring{\sigma})&=&\int d^dx\,[G_f(\bm{x})]^j\,e^{-i(\bm{p}\cdot\bm{r}+\bm{k}\cdot\bm{z} )}\nonumber\\&=&
 \begin{cases}\mathring{\sigma}^{-m/4}\,p^{-\epsilon}\,I_2(\mathring{\sigma}^{1/4}kp^{-1/2}),&j=2,\\[\smallskipamount]
 \mathring{\sigma}^{-m/2}\,p^{2-2\epsilon}\,I_3(\mathring{\sigma}^{1/4}kp^{-1/2}),&j=3.
\end{cases}
 \end{eqnarray}
 
Splitting off a factor $F_{m,\epsilon}^{j-1}$ from $I_j(\mathsf{k};m,d)\equiv I_j(\mathsf{k})$, let us write the Laurent expansions of the resulting ratios as 
 \begin{equation}
 \label{eq:I3Kexp}
\frac{I_j(\mathsf{k};m,d)}{F_{m,\epsilon}^{j-1}}=R^{(j)}_{-1}(\mathsf{k};m)\,\epsilon^{-1}+R^{(j)}_{0}(\mathsf{k};m)+R^{(j)}_{1}(\mathsf{k};m)\,\epsilon+\Or(\epsilon^2),\;\;j=2,3.
\end{equation}
In Appendix~\ref{app:I23} we determine the low-order coefficients. As residues we recover the results of Refs.~\cite{DS00a} and \cite{SD01}, namely 
 \begin{equation}
 \label{eq:R2m1}
R^{(2)}_{-1}(\mathsf{k};m)=1
\end{equation}
and
\begin{equation}
 \label{eq:R3m1}
R^{(3)}_{-1}(\mathsf{k};m)=\frac{j_\sigma(m)\,\mathsf{k}^4}{16 m(m+2)}-\frac{j_\phi(m)}{2(8-m)}.
\end{equation}
Following Ref.~\cite{SD01}, we have introduced here the integrals
\begin{equation}
 \label{eq:jsigma}
j_\sigma(m)=B_m\,\int_0^\infty \upsilon^{m+3}\,\Phi_{m,d^*}^3(\upsilon)\,d\upsilon
\end{equation}
and
\begin{equation}
 \label{eq:jphi}
j_\phi(m)=B_m\,\int_0^\infty \upsilon^{m-1}\,\Phi_{m,d^*}^3(\upsilon)\,d\upsilon
\end{equation}
with
\begin{equation}
 \label{eq:Bm}
B_m\equiv
\frac{S_{m-1}S_{d^*-m-1}}{F_{m,0}^2}=
\frac{2^{10+m} \pi ^{6+3 m/4}\, \Gamma(m/2)}{\Gamma(2-m/4) \,\Gamma^2(m/4)},
\end{equation}
where
\begin{equation}
 \label{eq:Smm1}
S_{m}=2\pi^{(m+1)/2}/\Gamma[(m+1)/2]
\end{equation}
is the volume of the unit $m$-sphere $S^{m}$. Unfortunately, analytic results are known for the integrals $j_\phi(m)$ and $j_\sigma(m)$ only for the special values $m=2$ and $6$, as well as for the limiting values
$j_\phi(0+)$ and $j_\sigma(8-)$; see Eqs.~(51)-(55),  (84), and (86) of Ref.~\cite{SD01}. For general choices of $m$, these integrals can be determined by numerical integration. Results can be found in Table~1 of Ref.~\cite{SD01}.

In Appendix~\ref{app:I23} we show that  the regular parts $R^{(2)}_{l\ge 0}(\mathsf{k}=0;m)$ vanish exactly at $\mathsf{k}=0$ [Eq.~\eqref{eq:I2zero}]. Hence 
\begin{equation}
 \label{eq:R20kzero}
R^{(2)}_l(0;m)=\delta_{l,-1}.
\end{equation}
The finite parts $R^{(2)}_0(\mathsf{k};m)$ and $R^{(3)}_0(\mathsf{k};m)$ can be expressed as double integrals involving the generalized functions $r_+^{-1}$ and $r_+^{-3}$, respectively. We have
\begin{equation}
\label{eq:R20res}
R^{(2)}_0(\mathsf{k};m)=\frac{2^{6+m/2} \pi ^{m+4}}{\Gamma(2+m/4)\, \Gamma(m/4)}\int_{-\infty}^\infty dr\,r_+^{-1}\,\varphi^{(2)}_{m,d^*}(r;\mathsf{k})
\end{equation}
and
\begin{eqnarray}
\label{eq:R30res}
R^{(3)}_{0}(\mathsf{k};m)&=&
\frac{1}{2}\,R^{(3)}_{-1}(\mathsf{k};m)\big[2\gamma_E-4 +\psi(2-m/4)-\ln(16\pi^3)\big]
\nonumber\\ &&\strut 
+B_m\int_{-\infty}^\infty dr\,r_+^{-3}\,\varphi^{(3)}_{m,d^*}(r;\mathsf{k}),
\end{eqnarray}
where $\psi(x)=d\ln\Gamma(x)/dx$ is the digamma function \cite{GR80}). Further, $\gamma_E=-\psi(1)$ is the Euler-Mascheroni constant and $\varphi^{(j)}_{m,d^*}(r;\mathsf{k})$ means the function
\begin{equation}
\varphi^{(j)}_{m,d}(r;\mathsf{k})=\,_0F_1\bigg(\frac{d-m}{2},-\frac{r^2}{4}\bigg)\int_0^\infty d\upsilon\,\upsilon^{m-1}\,\Phi_{m,d}^j(\upsilon)\,\,_0F_1\bigg(\frac{m}{2},-\frac{r\mathsf{k}^2\upsilon^2}{4}\bigg)
\end{equation}
with $d$ set to $d^*=4+m/2$. The reader may consult Refs.~\cite{GS64} and \cite[Appendix]{Die86a} for general background on the distributions $r_+^{-s+\epsilon}$ and their Laurent expansions. The definitions of the generalized functions  that are encountered in these expansions are given in Eq.~\eqref{eq:Laurdist} of Appendix~\ref{app:I23}. Suffice it here to recall the simple example of $r_+^{-1}$. This distribution acts on test functions $\varphi(r)$ as
\begin{equation}
 \label{eq:rm1def}
(r_+^{-1},\varphi)\equiv\int_{-\infty}^\infty dr \,r_+^{-1}\,\varphi(r)\equiv\int_0^1 dr\,r^{-1}[\varphi(r)-\varphi(0)]+\int_1^\infty dr\,r^{-1}\,\varphi(r).
\end{equation}

With the aid of the above results the renormalization functions $B(\mathring{u},\mathring{\sigma};\mu)$, $\mathcal{B}(u)$, and the renormalized functions $\tilde{G}^{(2,0)}_{\mathrm{R}}$ and $\tilde{G}^{(0,2)}_{\mathrm{R}}$ can be computed in a straightforward manner. Our results are
\begin{equation}
B(\mathring{u},\mathring{\sigma};\mu)=\mu^{-\epsilon}\mathring{\sigma}^{-m/4} \left[\frac{n}{2}\,\frac{F_{m,\epsilon}}{\epsilon}-\frac{\mathring{u}\,\mathring{\sigma}^{-m/4}}{\mu^\epsilon}\,\frac{n(n+2)}{12}\,\frac{F^2_{m,\epsilon}}{\epsilon^2}\right]+\Or(\mathring{u}^2),
\end{equation}
\begin{equation}
\mathcal{B}(u)=-F_{m,\epsilon}\,\frac{n}{2}+\Or(u),
\end{equation}
 \begin{eqnarray}
 \label{eq:Gamma20R}
1/\tilde{G}^{(2,0)}_{\mathrm{R}}(\bm{k},\bm{p})&=&p^2+\sigma k^4-\frac{n+2}{3}\,
\frac{u^2}{6}\,p^2{\Big[}{R^{(3)}_0}(k_p;m)
\nonumber\\ &&\strut -2\,R_{-1}^{(3)}(k_p;m)\ln(p/\mu)+\Or(\epsilon)\Big]+\Or(u^3),
\end{eqnarray}
and
\begin{eqnarray}
 \label{eq:G02R}
\tilde{G}^{(0,2)}_{\mathrm{R}}(\bm{k},\bm{p})&=&
\frac{n}{2}\,\frac{F_{m,\epsilon}}{\mu^{\epsilon}\,\sigma^{m/4}}
\left\{R_0^{(2)}(k_p;m)-\ln\frac{p}{\mu}+\epsilon R_1^{(2)}(k_p;m)\right.
\nonumber\\
&&\strut
-\frac{1}{2}\Big(\epsilon-u\frac{n+2}{3}\Big)\Big[2 R_0^{(2)}(k_p;m)-\ln\frac{p}{\mu}\Big]
\ln\frac{p}{\mu}\nonumber\\
&&\strut - u\frac{n+2}{6}\big[R_0^{(2)}(k_p;m)\big]^2\bigg\}+\Or(\epsilon^2,u^2,u\epsilon),
 \end{eqnarray}
where 
\begin{equation}
\label{eq:kp}
k_p\equiv\sigma^{1/4} k p^{-1/2}.
\end{equation}

Upon setting $u$ to its fixed-point value~\eqref{eq:ustar}, one can convince oneself that the results are in conformity with the scaling forms~\eqref{eq:RGsolG20} and \eqref{eq:RGsolG02}. The scaling functions $\Psi$ and $\Upsilon$ are found to have the $\epsilon$ expansions
 \begin{equation}
 \label{eq:Psires}
1/\Psi_{m,d^*-\epsilon}(\mathsf{k})=1+\mathsf{k}^4-2\epsilon^2\,\frac{n+2}{(n+8)^2}\,R^{(3)}_0(\mathsf{k};m)+\Or(\epsilon^3)
\end{equation}
and
 \begin{eqnarray}
 \label{eq:Yres}
\Upsilon_{m,d^*-\epsilon}(\mathsf{k})&=&1+\frac{\alpha_{\mathrm{L}}}{\nu_{\mathrm{L}2}}
\bigg\{R_0^{(2)}(\mathsf{k};m)
+\epsilon\,R_1^{(2)}(\mathsf{k};m)
-\epsilon\,\frac{n+2}{n+8}\,\big[R_0^{(2)}(\mathsf{k};m)\big]^2\Big\}
+\Or(\epsilon^3)\nonumber\\ &=&\left[1+R_0^{(2)}(\mathsf{k};m)\,\epsilon+R_1^{(2)}(\mathsf{k};m)\,\epsilon^2\right]^{\alpha_{\mathrm{L}}/(\nu_{\mathrm{L}2}\epsilon)}+\Or(\epsilon^3)
\end{eqnarray}
In deriving the first and second form of Eq.~\eqref{eq:Yres}, we made use of the fact that 
\begin{equation}\label{eq:alphanueps}
\frac{\alpha_{\mathrm{L}}}{\nu_{\mathrm{L}2}}=\frac{4-n}{n+8}\,\epsilon+\Or(\epsilon^2).
\end{equation}
The second form at this stage is just a convenient rewriting of the $\epsilon$~expansion from which we shall benefit in Section~\ref{sec:endensextrapol} when extrapolating the scaling function $\Upsilon_{m=1,d^*-\epsilon}$ to $d=3$ dimensions.

A cautionary remark is in order here.  Expansions of scaling functions in powers of $\epsilon$ such as those given in Eqs.~\eqref{eq:Psires} and the first line of Eq.~\eqref{eq:Yres}, as well as those derived below, are not directly suitable for extrapolations to $d=3$. They must be supplied with appropriate exponentiation hypotheses in order to capture the correct limiting behavior for $\mathsf{k}\to\infty$. In the case of the functions $\Psi_{m,d}(\mathsf{k})$ and $\Upsilon_{m,d}(\mathsf{k})$, this asymptotic large-$\mathsf{k}$ behavior is dictated by the requirement that the $p$-dependencies of $\tilde{G}^{(2,0)}_R(\bm{k},\bm{p})$ and $\tilde{G}^{(0,2)}_R(\bm{k},\bm{p})-\tilde{G}^{(0,2)}_R(0,\mu\hat{\bm{p}})$ drop out. A convenient possibility to achieve the exponentiation of the asymptotic power laws in the limits $\mathsf{k}\to\infty$ and $\mathsf{k}\to0$ is to choose the scale parameter $\ell\equiv \ell_1$ such that the dimensionless inverse free susceptibility $(p^2+\sigma k^4)/\mu^2$ becomes $1$ at the chosen scale $\ell_1$. Let us split off a factor $p/\mu$,  writing $\ell_1=\ell_{\mathsf{k}}p/\mu$. Then $\ell_{\mathsf{k}}$ must be a solution to
\begin{equation}\label{eq:lchoice}
\ell_{\mathsf{k}}^{-2}+\ell_{\mathsf{k}}^{-4\theta}\,\mathsf{k}^4=1.
\end{equation}
Instead of Eq.~\eqref{eq:scfPsi}, we then have
\begin{equation}\label{eq:Psilk}
\Psi_{m,d}(\mathsf{k})=\ell_{\mathsf{k}}^{\eta_{\mathrm{L}2}-2}\,\tilde{G}^{(2,0)}_R(\ell_{\mathsf{k}}^{-\theta}\mathsf{k},\hat{\bm{p}}\ell_{\mathsf{k}}^{-1};u^*,1,1)
\end{equation}
and a corresponding modification of Eq.~\eqref{eq:scfUps}, where 
\begin{equation}\label{eq:lkas}
\ell_{\mathsf{k}}\approx
\begin{cases}
1,&\text{for }\mathsf{k}\to0,\\
\mathsf{k}^{1/\theta},&\text{for }\mathsf{k}\to\infty.
\end{cases}
\end{equation} We shall return to the issue of the extrapolation of the scaling functions in Sec.~\ref{sec:endensextrapol} when extrapolating the $\epsilon$-expansion of $\Upsilon_{m=1,9/2-\epsilon}(\mathsf{k})$ to $d=3$.

\subsection{The special cases $d=m+3$, $d=m+3=5-2\epsilon$, and  $(m,d)=(2,5-\epsilon)$}\label{sec:corrfctdm3}

For general values of $m$, the functions $R^{(j)}_l$ appearing in the $\epsilon$ expansions~\eqref{eq:Psires} and \eqref {eq:Yres} of the scaling functions would have to be determined by numerical means. However, on the line  $d=m+3$,  the functions $\Phi_{m,d}(v)$ reduce to simple Gaussians,
\begin{equation}
   \label{eq:frm2}
   \Phi_{m,m+3}(v)=(4\pi)^{-(m+2)/2}\,e^{-v^2/4}.
 \end{equation}
This makes it possible to determine the functions $R^{(j)}_l(\mathsf{k};2)$ in closed analytical form. Details of the calculations are described in Appendix~\ref{app:I23}. To state the results and for subsequent use it is convenient to introduce the functions
\begin{equation}
\label{eq:Q0def}
Q_0(s)=\Im[(s-i) \ln(s-i)]=-s\,\vartheta(s)-\frac{1}{2}\ln(1+s^2)
 \end{equation}
with
\begin{equation}\label{eq:defvartheta}
\vartheta(s)=\arctan(1/s)
\end{equation}
and the dilogarithm $\mathrm{Li}_2(z)$, a special one of the polylogarithm functions,
defined by analytic continuation of their Taylor series (see Refs.~\cite{Lewin} and \cite[Sec.~2.6]{AAR}).
\begin{equation}
\mathrm{Li}_k(z)= \sum_{j=1}^\infty\frac{z^j}{j^k}\quad \text{for }|z|<1
\quad \text{and } k=2,3,\dotsc
\end{equation}
to the complex plane, with a branch cut along the positive real $z$-axis from $z=1$ to $z=\infty$.
The analytical continuation of $\mathrm{Li}_2(z)$ is provided by its
integral representation \cite{Lewin,AAR,Max03}
\begin{equation}
\mathrm{Li}_2(z)=-\int_0^zdt\,\frac{\ln(1-t)}{t}=-\int_0^1 dt\,\frac{\ln(1-zt)}{t}\,.
\end{equation}
We shall also need the Clausen integral (see Refs.~\cite{Lewin}, \cite{Max03}, and \cite[Sec.~27.8]{AS72})
\begin{equation}
\mathrm{Cl}_2(z)\equiv\sum_{j=1}^\infty \frac{\sin(jz)}{j^2}=
\Im [\mathrm{Li}_2(e^{iz})]=-\int_0^zdt\,\ln\Big[2\sin\frac{t}{2}\Big]\,.
\end{equation}
The above special functions are intensively used in papers dealing with
calculations of Feynman integrals appearing in usual $\phi^4$ theory and QFT
\cite{DevDuke84,DavDel98,Dav00,DavKal01,Coffey08}. % DavKal01 replaced by DavKal04

In terms of these quantities our results read
\begin{equation}
\label{eq:R20m2}
R^{(2)} _ 0 (\mathsf{k}; 2)=Q_0(\mathsf{k}^2/2)\,,
\end{equation}
\begin{eqnarray}
\label{eq:R21m2}
R_1^{(2)}(\mathsf k;2)&=&\frac{\pi^2}{6}-\frac{1}{2}
-\frac{\mathsf k^2}{4}\vartheta(\mathsf{k}^2/2)-\vartheta^2(\mathsf{k}^2/2)+\frac{\vartheta(\mathsf{k}^2/2)-\vartheta(\mathsf{k}^2)}{\mathsf k^{2}}\nonumber\\&&\strut 
-\frac{\mathsf{k}^2}{4}\,[\vartheta(\mathsf{k}^2/2)-\vartheta(\mathsf{k}^2)]\ln{\bigg(\frac{4+\mathsf k^4}{\mathsf{k}^4}\bigg)} +
\frac{\mathsf k^2}{2}\vartheta(\mathsf{k}^2/2)\ln\bigg(\frac{4+\mathsf k^4}{8}\bigg)\nonumber\\
\nonumber
&&\strut
+\frac{1}{2}\ln\bigg(\frac{4+4\mathsf k^4}{4+\mathsf k^4}\bigg)+
\frac{1}{8}\ln^2(4+\mathsf k^4)+\Re\bigg[\mathrm{Li}_2\Big(\frac{i}{2i+\mathsf k^2}\Big)\bigg]
\\\nonumber
&&\strut
-\frac{1}{4}\mathrm{Li}_2(-\mathsf{k}^4/4)
+\frac{1}{4}\mathrm{Li}_2(-\mathsf{k}^4)-\frac{\mathsf{k}^2}{2}\,\mathrm{Cl}_2{\big[2\vartheta(\mathsf{k}^2/2)\big]}\\
&&\strut
+\frac{\mathsf k^2}{4}\,\mathrm{Cl}_2\big[2\vartheta(\mathsf{k}^2)\big]
+\frac{\mathsf k^2}{4}\,\mathrm{Cl}_2\big[2\vartheta(\mathsf{k}^2/2)-2\vartheta(\mathsf{k}^2)\big],
\end{eqnarray}
and
\begin{equation}
\label{eq:R30m2}
R^{(3)} _ 0 (\mathsf{k}; 2)=\frac{4}{18} \bigg(\frac{11}{6}-\gamma_E\bigg)
\Im \bigg[\bigg(\frac{\mathsf{k}^2}{3}+i\bigg)^3\bigg]-\frac{2}{9}\,
\Im\bigg[\bigg(\frac{\mathsf{k}^2}{3}+i\bigg)^3
\ln\bigg(\frac{\mathsf{k}^2}{3}+i\bigg)\bigg].
\end{equation}

Furthermore, it is possible to compute the required Feynman integrals on the line $d=m+3$ in closed form (see Appendix~\ref{app:I23}) for general values of $m$. When $d=d^*(m)-\epsilon$, the constraint $d=m+3$ implies that $d=5-2\epsilon$ and $m=2-2\epsilon$. Using the results for the integrals $I_j(\mathsf{k};m,m+3)$ given in Eqs.~\eqref{eq:I2deqmp3} and \eqref{eq:I3deqmp3}, one can compute the $\epsilon$ expansions of the scaling functions $\Psi_{2-2\epsilon,5-2\epsilon}$ and $\Upsilon_{2-2\epsilon,5-2\epsilon}$. Our results are
\begin{eqnarray}
\label{eq:Psidm3exp}
1/\Psi_{2-2\epsilon,5-2\epsilon}(\mathsf{k})&=&1+\mathsf{k}^4+2\epsilon^2\,\frac{n+2}{(n+8)^2}\Bigg\{\frac{1}{27}\bigg(5-\ln\frac{64}{27}\bigg)
\left(1-\frac{\mathsf{k}^4}{3}\right)\nonumber\\&&\strut+\frac{2}{9}\,\Im\bigg[\bigg(\frac{\mathsf{k}^2}{3}+i\bigg)^3\ln\bigg(\frac{\mathsf{k}^2}{3}+i\bigg)\bigg]\Bigg\}+\Or(\epsilon^3)
\end{eqnarray}
and
\begin{equation}
\label{eq:Ydm3exp}
\Upsilon_{2-2\epsilon,5-2\epsilon}(\mathsf{k})=1+\epsilon\frac{4-n}{n+8}
\Big\{Q_0(\mathsf{k}^2/2)
+\epsilon\, Q_1(\mathsf{k}^2/2)
-\epsilon\frac{n+2}{n+8}\,\big[Q_0(\mathsf{k}^2/2)\big]^2\Big\}
+\Or(\epsilon^3),
\end{equation}
where $Q_1(s)$ is defined by
\begin{eqnarray}
\label{eq:Q1def}
Q_1(s)&\equiv&\frac{\pi^2}{8}+\frac{1}{2}\,\Im[(s+i) \ln^2( s+i)]\nonumber\\
&=&\frac{\pi^2}{8}+\frac{1}{8}\ln^2(1+s^2)+
\frac{s}{2}\vartheta(s)\ln(1+s^2)-\frac{1}{2}\vartheta^2(s).
\end{eqnarray}

In Fig.~\ref{fig:Qs}, %
\begin{figure}[htbp]
\begin{center}
\includegraphics[width=90mm]{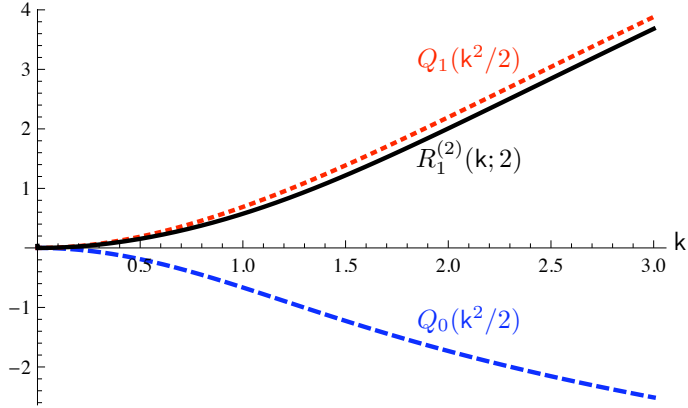}
\caption{The functions $Q_0(\mathsf{k}^2/2)=R_0^{(2)}(\mathsf{k};2)$ (blue, dashed), $R_1^{(2)}(\mathsf{k};2)$ (black, full line), and $Q_1(\mathsf{k}^2/2)$ (red,dotted).}
\label{fig:Qs}
\end{center}
\end{figure}
the functions  $Q_0(\mathsf{k}^2/2)$ [Eq.~\eqref{eq:Q0def}], $R_1^{(2)}(\mathsf{k};2)$ [Eq.~(\ref{eq:R21m2})], and $Q_1(\mathsf{k}^2/2)$ [Eq.~\eqref{eq:Q1def}] are plotted as blue dashed, black full, and red dotted curves, respectively. Note that the  $\Or(\epsilon)$ terms of the scaling functions $\Upsilon_{2,5-\epsilon}(\mathsf{k})$ [Eq.~\eqref{eq:Yres}] and $\Upsilon_{2-2\epsilon,5-2\epsilon}(\mathsf{k})$ [Eq.~\eqref{eq:Ydm3exp}] agree and are given by $\epsilon \,Q_0(\mathsf{k}^2/2)\,(4-n)/(n+8)$. The function $Q_1(\mathsf{k}^2/2)$ is the analog of $R^{(2)}_1(\mathsf{k};2)$ for the case $d-m=3$. Their difference is small.

\subsection{The energy-density cumulant in the uniaxial case $m=1$}
    \label{sec:endens} 
 
Since the predictions of the LSI theory were used to fit the Monte Carlo data for the ANNNI model, the case of a uniaxial Lifshitz point $m=1$ is of particular interest. Using our results for general values of $m$ described in Section~\ref{sec:genm}, we could determine the scaling functions for $m=1$ by numerical integration. However, in order to compare with the predictions of the LSI theory it is preferable to have as much precise mathematical knowledge available as possible. It turns out that more detailed valuable analytical results can be obtained for the scaling function $\Upsilon$ of the energy-density cumulant. To do this we will start from the  momentum-space representation of the Feynman integral $I_2(\mathsf{k};1,d)$, derive a contour integral representation of the latter, and relate it to solutions of a Fuchsian third order differential equation.

The scaling properties of the integral $J_2(\bm{k},\bm{p};\sigma)$ defined in Eq.~\eqref{eq:Jldefx} imply that  the function $I_2(\mathsf{k};m=1,d)$ can be written as
\begin{equation}
\label{eq:I2}
I_2(\mathsf{k};1,d)=\mathsf{k}^{-2\epsilon} J(\mathsf{k}^{-2})
\end{equation}
with
   \begin{equation}
\label{eq:JPdef} 
J(P)=\int_{-\infty}^\infty\frac{dk}{2\pi}\int \frac{d^{d-1} p}{(2 \pi )^{d-1}}\, \frac{1}{(p^2+ k^4)
[({\bm p} - \bm{P})^2+(k-1)^4]}.
\end{equation}
 
The integral on the right-hand side of this equation  converges for $5/2<d<9/2$. It defines a function of $P$ that is regular near the origin $P=0$. Hence it has a Taylor expansion of the form
 \begin{equation}
 \label{eq:JPsm}
 J(P)=\sum_{j=0}^\infty C_j P^{2 j}. 
 \end{equation}
 On the other hand, the function $P^\epsilon J(P)$ is regular at  $P = \infty$.  Therefore, $J(P)$ can be expanded as 
 \begin{equation}
 \label{eq:JPlg}
   J (P)=P^{-\epsilon}\sum_{j=0}^\infty B_j P^{- j}
 \end{equation}
 for large $P$.

Much more information can be gained from the following contour-integral representation of $J(P)$ proved in Appendix~\ref{app:Cont}:
  \begin{equation}
 \label{eq:JcalJ}
  J (P)= \frac{S_{d-3}}{(2\pi)^{d-2}} \,\frac{2^{4-2d}} {\cos(\pi\, d)}\,\mathcal{J}(P^2)
 \end{equation}
with
\begin{equation}
 \label{eq:calJw}
\mathcal{J}(w)= -i\Big[ \mathcal{J}_1 (w)+\big(1-e^{-i\pi d}\big)\,\mathcal{J}_3(w)\Big],
\end{equation}
where 
 \begin{equation}
 \label{eq:Jcal1}
 \mathcal{J}_1 (w) = \int_{t_-(w)}^{t_+(w)} \frac{dt}{\sqrt{t}}(t-1)^{(d-4)/2} \rho (t,w)^{(d-4)/2}
 \end{equation}
 and
 \begin{equation}
  \label{eq:Jcal3}
 {\mathcal J}_3(w) =e^{i\pi d/2}\int_{0}^{1} \frac{dt}{\sqrt{t}}(1-t)^{(d-4)/2} \rho (t,w)^{(d-4)/2}. 
 \end{equation}
 Here $\rho(t,w)$ denotes the function
  \begin{equation}
  \label{eq:rho}
 \rho(t,w)= (w t-1)^2+4w
 \end{equation}
 and 
 \begin{equation}
 \label{eq:tpm}
   t_{\mp}(w)=w^{-1}\mp 2 i \,w^{-1/2}
   \end{equation}
 are its zeros.

The integration paths of the integrals~\eqref{eq:Jcal1} and \eqref{eq:Jcal3} are illustrated in Fig.~\ref{fig:B0int}.
\begin{figure}[htbp]
\begin{center}
\includegraphics[width=0.45\textwidth]{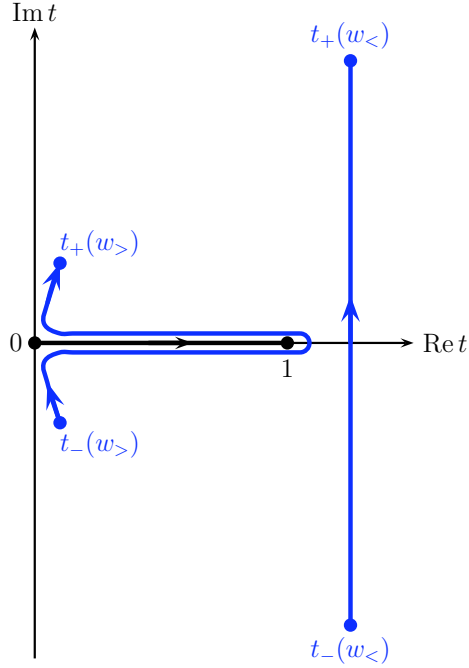}
\caption{Integration paths for the integrals $\mathcal{J}_1(w)$ (black), $\mathcal{J}_3(w_<)$ (blue), and $\mathcal{J}(w_>)$ (blue) with $0<w_<<1$ and $1<w_><\infty$, respectively.}
\label{fig:B0int}
\end{center}
\end{figure}
For values of $w$ with $0<w<1$, the path for $\mathcal{J}_1(w)$ is parallel to the imaginary axis. When $1<w<\infty$,  the roots~\eqref{eq:tpm} have real parts $\Re \,t_\pm=1/w<1$. To define $\mathcal{J}_1(w)$ in this case by proper analytic continuation, a path must be chosen that goes around the branch cut $(0,1)$ of its integrand. 
Deforming this path such that the subpaths along the $\Im\, t<0$ and $\Im \, t>0$ rims of the branch cut pass through $\pm i0$, respectively, one sees that the sum of integrals along these subpaths cancels the contribution to $\mathcal{J}(w)$ from $\mathcal{J}_3(w)$. Hence we can rewrite $\mathcal{J}(w)$ as
\begin{equation}
\label{eq:Jcalrewritten}
\mathcal{J}(w)=-i\left(\int_{t_-(w)}^{-i0}+\int_{i0}^{t_+(w)}\right)\frac{dt}{\sqrt{t}}\,[(t-1)\,\rho(t,w)]^{(d-4)/2} \quad\text{when }0<w<\infty.
\end{equation}

The integrals~\eqref{eq:Jcal1} and \eqref{eq:Jcal3}  as well as \eqref{eq:Jcalrewritten} converge for $d>2$. The UV singularities of $J(P)$ therefore originate from the zeros of the factor  $\cos (\pi d)=\sin(\pi \epsilon)$ in the denominator of the right-hand side of Eq.~\eqref{eq:JcalJ}. The integrals~\eqref{eq:Jcal1} and \eqref{eq:Jcal3} are of the Euler type. They belong to the same class as  Euler's hypergeometric integral 
\begin{equation} \label{eq:f21int}
_{2}F_1(a,b;c;z)=\frac{\Gamma(c)}{\Gamma(b)\,\Gamma(c-b)}\int_0^1\frac{dw}{w^{1-b}\,(1-w)^{1-c+b}\,(1-wz)^a}
\end{equation}
defining the hypergeometric function $_{2}F_1$ [see e.g.\ Eq.~(15.3.1) of Ref.~\cite{AS72}]. There are several ways to show that the so-defined function is a solution to the hypergeometric differential equation \cite{Mas97}. We pursue a similar strategy here. Following the procedure described in  Masaaki's book  \cite[pp. 88--89]{Mas97}, one can derive an ordinary differential equation of third order and Fuchsian type \cite{Fuc1866,Fuc1868} that is solved by the Euler integrals~(\ref{eq:Jcal1}) and (\ref{eq:Jcal3}). Such differential equations are frequently exploited in studies of Feynman integrals and their singularities, see for example \cite{Bou07}.

The functions ${\mathcal J}(w)$, $\mathcal{J}_1(w)$, and $\mathcal{J}_3(w)$ all obey the same differential equation. Let us introduce the operator
\begin{equation}
\label{eq:Ddef}
D=w\frac{d}{dw},
\end{equation}
along with the parameter
\begin{equation}
\label{eq:lambdadef}
\lambda=(d-4)/2
\end{equation}
and the coefficients
\begin{equation}
\label{eq:a0def}
 a_0 =-\frac{\lambda\,(1+6 \lambda)(8\lambda-9)}{6(1+4 w)}- \frac{\lambda\,(3+4\lambda)(1+6\lambda)}{6(1+ w)},
\end{equation}
\begin{equation}
\label{eq:a1def}
 a_1 =\frac{392 \lambda-39-304 \lambda^2}{48(1+4 w)}- \frac{\lambda\,(52\lambda-5)}{6(1+ w)}- \frac{(4\lambda-3)(4\lambda-1)}{16}, 
\end{equation}
and 
\begin{equation}
\label{eq:a2def}
 a_2=\frac{3-2\lambda}{2(1+4 w)}- \frac{4 \lambda}{1+ w}+2-2\lambda.
\end{equation}
Then this differential equation can be written as
 \begin{eqnarray}
 \label{eq:dif}
 \lefteqn{(2-2 \lambda+D)(1-2 \lambda+D)(D-2 \lambda) \mathcal{J}(w)}&&\nonumber\\ &=& 
 a_0 \,\mathcal{J}(w)+a_1 (D-2 \lambda) \mathcal{J}(w)+
 a_2\,(1-2 \lambda+D)(D-2 \lambda) \mathcal{J}(w).
 \end{eqnarray}

Inspection of the coefficients of its terms proportional to $d^k\mathcal{J}/dw^k$ with $k=3,2,1,0$ reveals that it is a Fuchsian differential equation with regular singular points at $w=0$, $w=-1/4$, $w=-1$, and $w=\infty$. It has three linearly independent solutions, two of which are regular at the origin. The pole of the coefficients $a_0$ and $a_1$ that is closest to the origin is located at $w=-1/4$. Hence the Taylor expansions of solutions ${\mathcal J}(w)$ that are regular at the origin, 
 \begin{equation}
  \label{eq:Taylor}
  \mathcal{J}(w) =\sum_{j=0}^\infty \, A_jw^j, %subscript reg on \mathcal{J}_{\text{reg}} removed
 \end{equation}
are guaranteed to converge inside the disc $|w|<1/4$. 
Substituting this expansion into Eq.~\eqref{eq:dif} leads to the recursion relations 
 \begin{eqnarray}
   \label{eq:rec}
A_{j+2} &=&\frac{(4 \lambda-1-4 j)(3+4 j-4 \lambda)(j-2\lambda)}{2(1+j)(2+j)(5+2 j+2 \lambda)} \,A_j
 \nonumber\\ &&\strut 
 - \frac{35+54\, j+20\,j^2-26\lambda-12j\lambda-8\lambda^2}{2(2+j)(5+2 j+2 \lambda)}A_{j+1}. 
 \end{eqnarray}

The low-order coefficients $A_0$ and $A_1$ can be computed in a straightforward fashion from Eqs.~\eqref{eq:calJw}--\eqref{eq:Jcal3}. The  change of variables $t=w^{-1}+2iw^{-1/2}v\to v$ transforms  $ \mathcal{J}_1 (w)$ into an integral from $v=-1$ to $v=+1$. Expanding the integrand in powers of $\sqrt{w}$ to order $w$ and integrating term by term then gives
\begin{equation}
 \mathcal{J}_1 (w)=i\,2^{1+2\lambda}\,N_\lambda\left[1-\frac{3-2\lambda+8\lambda^2}{2(3+2\lambda)}\,w\right]+ \Or(w^2)
\end{equation}
with
\begin{equation}\label{eq:Nlambda}
N_\lambda=\pi^{1/2}\,\Gamma(\lambda+1)/\Gamma(\lambda+3/2).
\end{equation}
The expansion
\begin{equation}
 \mathcal{J}_3 (w)=e^{i\pi\lambda}\,N_\lambda\left[1+\frac{2\lambda(5+4\lambda)}{3+2\lambda}\,w\right]+ \Or(w^2)
\end{equation}
can be obtained in a similar fashion. Writing the Taylor expansion of $\mathcal{J}(w)$ as 
 \begin{equation}\label{eq:TEX}
 \mathcal{J}(w) = N_\lambda\,
 \sum_{j=0}^\infty \mathcal{A}_j w^j,
 \end{equation}
we can combine these results with Eq.~\eqref{eq:calJw} to conclude that
\begin{equation}\label{eq:ANUL}
\mathcal{A}_0= 2[4^\lambda+\sin(\pi\lambda)]
\end{equation}
and
\begin{equation}\label{eq:AONE}
  \mathcal{A}_1= 4^{\lambda}\,\frac{2\lambda-8\lambda^2-3}{3+2\lambda}+
 4\lambda \,\frac{5+4\lambda}{3+2\lambda}\,\sin(\pi\lambda).
 \end{equation}
 
Substituting expression~\eqref{eq:ANUL} for $\mathcal{A}_0$ along with  Eqs.~\eqref{eq:Nlambda} and \eqref{eq:lambdadef} into Eq.~\eqref{eq:JcalJ} yields
\begin{equation}
C_0=J(0)=2^{4-3 d} \pi^{(3-d)/2}\,\frac{ 2^d+16
   \sin(\pi  d/2)}{\Gamma[(d-1)/2]\, \cos (\pi  d)}
   \end{equation}
for the coefficient $C_0$ appearing in Eq.~\eqref{eq:JPsm}. The result agrees with the value of $I(0,1)$ given by Eq.~(B.23) of Ref.~\cite{SPD05} when $m=1$.

The coefficients $\mathcal{A}_j$ with $j\geq 2$  can be determined, on the one hand, from $\mathcal{A}_0$ and $\mathcal{A}_1$ with the aid of the recursion relations~\eqref{eq:rec}. On the other hand, a general explicit expression for
them can be derived from Eq.~(26) of  Ref.~\cite{Sh07}, where the result of
the inner $p$ integration of the double integral~\eqref{eq:JPdef} for $J(P)$ is given
in terms of the Appell \cite{Appell26} function $F_1$. The parameter $\varepsilon$  used in Ref.~\cite{Sh07} is defined as $ \varepsilon=4-D$, where $D$ is the dimension of the $\bm{p}$-integral. Hence, we must set $\varepsilon =4-(d-1)=5-d=1-2\lambda$, $p_x=P$ and $m_{1,2}=(k\mp 1/2)^2$ in the relevant Eq.~(26) of this reference.  Substituting the standard representation of the Appell function $F_1$ as a double series, given in Eq.~(20) of Ref.~\cite{Sh07},  into Eq.~\eqref{eq:JPdef} one can integrate term by term over $k$ to determine the series expansion~\eqref{eq:JPsm} of the integral $J(P)$. The series expansion~\eqref{eq:TEX} of $\mathcal{J}(P^2)$ then follows via Eq.~\eqref{eq:JcalJ}. Our results for the coefficients are
\begin{eqnarray}\label{eq:AGEN}
\mathcal A_j&=&\frac{2\sin(\pi\lambda)}{\lambda\sqrt\pi}\,
\frac{\Gamma(2+2\lambda)}{\Gamma(1/2+2\lambda)}\,
\frac{(1/2-2\lambda)_{2j}}{(1-\lambda)_j(3/2)_j}\;(-1)^j
\nonumber\\&&\strut\times \,_3F_2(-\lambda,1/2,1/2-2\lambda+2j;
1-\lambda+j,3/2+j;1),
\end{eqnarray}
where $(c)_j\equiv\Gamma(c+j)/\Gamma(c)$ denotes the Pochhammer symbol. For $j=0$ and $j=1$, the last expression reproduces Eqs.~\eqref{eq:ANUL} and (\ref{eq:AONE}), respectively. Using {\sc Mathematica} \cite{Mathematica7}, one can check that these coefficients, Eq.~\eqref{eq:AGEN}, satisfy the recursion relations~\eqref{eq:rec}.

From Eq.~\eqref{eq:AGEN} one can obtain the $\epsilon$~expansion of the coefficients $\mathcal A_j$
to $\Or( \epsilon)$ in a simple manner. Recalling that $ \epsilon=1/2-2\lambda$ and
taking into account that $( \epsilon)_{2k}= \epsilon(2k-1)!+\Or( \epsilon^2)$
as $ \epsilon\to 0$, one finds
\begin{equation}\label{eq:AGON}
\mathcal A_0=3\sqrt 2-\sqrt 2(\pi/2+2\ln 2)\, \epsilon+\Or( \epsilon^2)
\end{equation}
and
\begin{equation}
\mathcal A_j= \epsilon\,\frac{3\sqrt 2 (2 j-1)!}{(3/4)_j(3/2)_j}\;(-1)^j
%\nonumber\\&&\strut
\,{_3F_2}(-1/4,1/2,2j;3/4+j,3/2+j;1)+\Or( \epsilon^2),\quad j\ge 1.
\end{equation}
The $\epsilon$~expansions of the first two coefficients $\mathcal A_j$ with $j\ge 1$ are given by
\begin{equation}\label{eq:A1}
\mathcal A_1=-\frac{3\sqrt 2}{7}(4+\pi-2\ln 2) \epsilon+\Or( \epsilon^2)
\end{equation}
and
\begin{equation}\label{eq:A2}
\mathcal A_2=\frac{3\sqrt 2}{11}\left[9+2(\pi-2\ln 2)\right] \epsilon
+\Or( \epsilon^2)\,.
\end{equation}
Higher-order terms of their $ \epsilon$ expansions can be obtained by means of the methods
developed in Refs.~\cite{Wei04} and \cite{DavKal04}.

Next, let us consider the asymptotic behavior of $\mathcal{J}(w)$ for  $w\to\infty$. As is clear from Eqs.~\eqref{eq:I2} and \eqref{eq:JcalJ}, the limiting form of  $\mathcal{J}(w\to \infty)$ is needed to determine the value of the function $I_2(\mathsf{k};1,d)$ at $\mathsf{k}=0$ and check its consistency with the previously obtained result for $I_2(0;m,d)$ \cite{SD01} recalled in Eq.~\eqref{eq:I2zero}. The required information can be derived by analytic means directly from the integral representation given in Eqs.~\eqref{eq:JcalJ}--\eqref{eq:Jcal3}. To see this, note that Eq.~\eqref{eq:JPlg} translates into a large-$w$ expansion of the form
\begin{equation}\label{eq:Jcallargewexp}
  {\mathcal J}(w)=  w^{\lambda-1/4}\sum_{j=0}^\infty {\mathcal B}_j \,w^{-j/2}.
\end{equation}
To determine the asymptotic term $\propto \mathcal{B}_0$, we substitute the approximation $t_\pm(w)\approx \pm 2i w^{-1/2}$ into the integral~\eqref{eq:Jcalrewritten}, obtaining
\begin{eqnarray}
\mathcal{J}(w)&\mathop{\approx}_{w\to\infty}&-iw^{2\lambda}\left(e^{i\pi\lambda}\int_0^{2iw^{-1/2}}+e^{-i\pi\lambda}\int^0_{-2iw^{-1/2}}\right)dt\,t^{-1/2}\,\bigg(t^2
-\frac{4}{w}\bigg)^\lambda \nonumber\\ &=&w^{\lambda-1/4}\,2^{2\lambda+3/2}\cos[\pi(\lambda-1/4)]\int_0^1dx\,x^{-1/2}(1-x^2)^\lambda.
\end{eqnarray}
The remaining integral can be performed. We thus arrive at the result
 \begin{equation}\label{eq:Bcal0}
  \mathcal{B}_0=2^{1/2+2\lambda} \cos[\pi(\lambda-1/4)]
 \frac{\Gamma(1/4)\,\Gamma(1+\lambda)}{\Gamma(\lambda+5/4)}.
 \end{equation}
 This can be combined with Eqs.~\eqref{eq:JcalJ} and \eqref{eq:I2} to compute $I_2(0;1,d)$. The result agrees indeed with Eq.~\eqref{eq:I2zero}.
 
 \subsection{Extrapolation of the $m=1$ energy-density scaling function to $d=3$}
 \label{sec:endensextrapol}
 
Before turning to a more detailed discussion of the above results, we will first use them to obtain an extrapolation of the energy-density scaling function $\Upsilon_{1,d}(\mathsf{k})$ to $d=3$ dimensions. Starting from the second form of Eq.~\eqref{eq:Yres}, we combine Eqs.~\eqref{eq:I2}, \eqref{eq:JcalJ}, \eqref{eq:Jcallargewexp}, and \eqref{eq:Bcal0} to conclude that its term in square brackets can be written as
\begin{equation}\label{eq:kJcalexp}
\mathsf{k}^{-2\epsilon}\,\mathcal{J}(\mathsf{k}^{-4})/\mathcal{B}_0=1+R_0^{(2)}(\mathsf{k};1)\,\epsilon+R_1^{(2)}(\mathsf{k};1)\,\epsilon^2+\Or(\epsilon^3).
\end{equation}
Substitution  of  the left-hand side into Eq.~\eqref{eq:Yres} yields the asymptotic large-$\mathsf{k}$ behavior $\mathsf{k}^{-2 \alpha_{\mathrm{L}}/\nu_{\mathrm{L}2}}$. The exponent may be recognized as the expansion of $-\alpha_{\mathrm{L}}/(\nu_{\mathrm{L}2}\theta)$  to $\Or(\epsilon)$. Hence, the $\epsilon$-expansion result~\eqref{eq:Yres} is consistent with the expected $k$-dependency of $\tilde{G}^{(0,2)}_R$ to this order.

In order to cast our $\epsilon$-expansion result in a form that  is well suited for extrapolating it to $d=3$, we follow the strategy which led to the crossover scaling form~\eqref{eq:Psilk} of $G^{(2,0)}_R$. We exploit the behavior of $\tilde{G}^{(0,2)}_R(\bm{\mathsf{k}},\hat{\bm{p}};0,0,u^*,1,1)$ under scale transformation, choosing the  scale parameter $\ell_{\mathsf{k}}$ again as the solution to Eq.~\eqref{eq:lchoice}. The analog of Eq.~\eqref{eq:Psilk} leads to the form 
\begin{equation}
 \label{eq:Yextrap}
D_1\Upsilon_{1,3}(D_2\mathsf{k})=\left[\frac{\mathcal{J}\big(\mathsf{k}^{-4}\ell_{\mathsf{k}}^{4\theta-2}\big)}{\mathcal{B}_0\,\mathsf{k}^{2\epsilon}\ell_{\mathsf{k}}^{\epsilon(1-2 \theta)}}\right]^{\alpha_{\mathrm{L}}/(\nu_{\mathrm{L}2}\epsilon)},
\end{equation}
where $D_1$ and $D_2$ are two nonuniversal metric factors we introduced  to adjust the amplitude of $\Upsilon_{1,3}$ and the scale of the variable $\mathsf{k}$. We have this freedom since we know from Eq.~\eqref{eq:RGsolG02} that the scaling form of $\tilde{G}^{(0,2)}_R$ involves two such factors (related to $E_\tau^*$ and $E_\sigma^*$). 
We will fix them via the normalization conditions
\begin{equation}\label{eq:normY1}
\Upsilon_{1,3}(0)=1
\end{equation}
and 
\begin{equation}\label{eq:normY2}
\lim_{\mathsf{k}\to\infty}\mathsf{k}^{\alpha_{\mathrm{L}}/(\nu_{\mathrm{L}2}\theta)} \Upsilon_{1,3}(\mathsf{k})=1.
\end{equation}
For $\mathcal{B}_0$, we substitute its exact ($d=3$)-value $\mathcal{B}_0|_{\lambda=-1/2}$. Further, we set $\epsilon=3/2$. Taking into  account that $\mathcal{J}(0)=N_{\lambda}\mathcal{A}_0$, $\mathcal{J}(w)\approx \mathcal{B}_0w^{-\epsilon/2}$ as $w\to\infty$, and the limiting behaviors~\eqref{eq:lkas} of $\ell_{\mathsf{k}}$, one finds from the normalization conditions:
\begin{eqnarray}
D_1&=&1,\nonumber\\ 
D_2&=&\left[N_\lambda\mathcal{A}_0/\mathcal{B}_0\right]^{-2\theta/3}\big|_{\lambda=-1/2}=\left[2\sqrt{\pi}\,\Gamma(3/4)/\Gamma(1/4)\right]^{-2\theta/3}.
\end{eqnarray}

For the critical exponents Eq.~\eqref{eq:Yextrap} involves, we use the $d=3$~estimates $\alpha_{\mathrm{L}}/\nu_{\mathrm{L},2}\simeq 0.21$ and $\theta\simeq 0.47$  of Ref.~\cite{SD01}. The function $\mathcal{J}(w)$ can be computed numerically from the integral representation 
\begin{equation}\label{eq:Jcalintrepfinal}
\mathcal{J}(w)=2\,\, \Im \int_{i0}^{t_+(w)}dt\,t^{-1/2}(t-1)^{-1/2}\,\rho(t,w)^{-1/2} \quad\text{for }0<w<\infty,
\end{equation}
to which Eqs.~\eqref{eq:calJw}--\eqref{eq:tpm} simplify when $d=3$.  Likewise, we use numerical means to solve Eq.~\eqref{eq:lchoice} for $\ell_{\mathsf{k}}$. The resulting scaling function $\Upsilon_{1,3}(\mathsf{k})$ one obtains in this fashion from Eq.~\eqref{eq:Yextrap} is depicted in Fig.~\ref{fig:Yextrap} and compared with the scaling function  $\Upsilon_{m=1,d=3}^{\text{free}}(\mathsf{k})$ of a free massless field theory with action
\begin{equation}\label{eq:Hfree}
\mathcal{H}_{m,d}^{\text{free}}[\mathcal{E}]=\frac{1}{2}\int \frac{d^mk}{(2\pi)^m}\int \frac{d^{d-m}p}{(2\pi)^{d-m}} \,(p^2+k^4)^{\tilde{\Delta}}\,\mathcal{E}_{\bm{k},\bm{p}}\mathcal{E}_{-\bm{k},-\bm{p}}\,.
\end{equation}
\begin{figure}[htbp]
\begin{center}
\includegraphics[width=0.6\columnwidth]{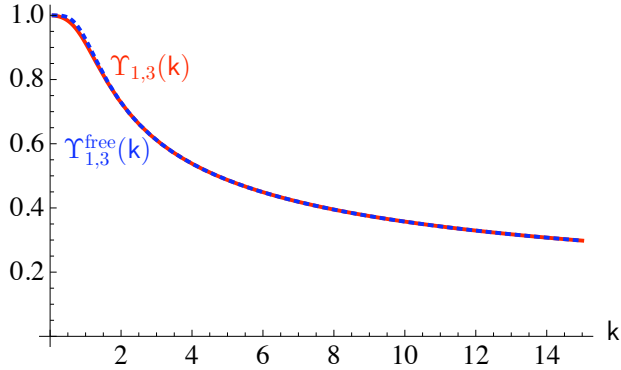}
\caption{Comparison of the extrapolated scaling function  $\Upsilon_{1,d=3}(\mathsf{k})$ (red full line) defined in Eq.~\eqref{eq:Yextrap} with the LSI prediction $\Upsilon^{\text{free}}_{1,3}(\mathsf{k})$ (blue dotted line) given in Eq.~\eqref{eq:YLSI}. For further explanations, see main text.}
\label{fig:Yextrap}
\end{center}
\end{figure}
The two-point cumulant one obtains from this action, 
\begin{eqnarray}\label{eq:YLSI}
\langle\mathcal{E}_{k',\bm{p}'}\mathcal{E}_{k,\bm{p}}\rangle^{\text{cum}} &=&p^{-2\tilde{\Delta}}\,\Upsilon_{m,d}^{\text{free}}(k p^{-\theta_4})\,(2\pi)^d\,\delta(\bm{k}+\bm{k}')\,\delta(\bm{p}+\bm{p}')\,,\nonumber\\
\Upsilon_{m,d}^{\text{free}}(\mathsf{k})&=&(1+\mathsf{k}^4)^{-\tilde{\Delta}}\,,
\end{eqnarray}
agrees with what the LSI prediction for $N=4$ ($\theta_N=1/2$) and $n=m=1$ given in Eqs.~\eqref{eq:modGscf} and \eqref{eq:Solg} yields  upon normalization according to Eqs.~\eqref{eq:normY1} and \eqref{eq:normY2} if  the unacceptable contribution $\propto b^{(4)}_1$ is dropped. The exponent $2\tilde{\Delta}$ corresponds to $\alpha_{\mathrm{L}}/\nu_{\mathrm{L}2}$. We therefore used the above-mentioned $(d=3,m=1)$-estimate $\alpha_{\mathrm{L}}/\nu_{\mathrm{L}2}\simeq 0.21$. As one sees from Fig.~\ref{fig:Yextrap}, the difference between our extrapolation $\Upsilon_{1,3}$ and $\Upsilon_{1,3}^{\text{free}}$ is fairly small ---  the two functions differ by at most $ 2\%$.

Finally, let us turn to a comparison of our  $\epsilon$-expansion results for the scaling function $\Upsilon_{1,9/2-\epsilon}^{\text{free}}$ with Henkel's prediction. Consider first the situation when $b_1^{(4)}=0$. In this case the LSI prediction is given by $\Upsilon_{1,9/2-\epsilon}^{\text{free}}$. To obtain the $\epsilon$~expansion of this function the result~\eqref{eq:alphanueps} with $n=1$  must be substituted for $2\tilde{\Delta}$. The 
$\Or(\epsilon)$ contribution of this function originates exclusively from the $\Or(\epsilon)$ term of  this exponent. Hence the  $\epsilon$~expansion of this scaling function clearly differs at order $\epsilon$ from our result given by Eqs.~\eqref{eq:Yres}, \eqref{eq:kJcalexp}, and \eqref{eq:Jcalintrepfinal}. Next, consider the LSI prediction with $b_1^{(4)}\ne 0$. Since the contribution of order $\epsilon^0$ is given by the Gaussian result $\Upsilon_{1,9/2}^{\text{free}}$, the coefficient $b_1^{(4)}$ would have to be of order $\epsilon$. The additional $\Or(\epsilon)$ term of the LSI scaling function resulting from the contribution $\propto b_1^{(4)}$ is incompatible with our $\epsilon$-expansion result. In fact, this incompatibility is not only quantitative but qualitative: The LSI scaling function $g(w)$ given in Eq.~\eqref{eq:Solg} has a single nontrivial singular point located at $w=1/4$. By contrast, the function $\mathcal{J}(w)$, which our $\Or(\epsilon)$ contribution involves, was found to have additional singularities, namely branching points.

Let us also show that our result for the energy-density correlation function $\tilde{G}^{(0,2)}_{\mathrm{R}}(\bm{k},\bm{p})$ given in  Eq.~\eqref{eq:RGsolG02}, unlike Henkel's $b_1^{(4)}\ne0$ LSI result, has expansions in $p^2$ when $k>0$ and in $k^2$ when $p>0$ of the forms~\eqref{eq:pexp} and \eqref{eq:kexp}, respectively. To this end, we return to Eq.~\eqref{eq:RGsolG02} and set $\mu=\sigma=E^*_\tau=E^*_\sigma=1$ for notational convenience. Our result~\eqref{eq:RGsolG02} then becomes
\begin{equation}
 \label{4.24}
\tilde{G}^{(0,2)}_{\mathrm{R}}(\bm{k},\bm{p}) \approx 
\mathsf{C}_1\,p^{-\alpha_L/\nu_{L2}}
\Upsilon_{m,d}(\mathsf{k}) +\mathsf{C}_2
\end{equation}
with $\mathsf{k}=k/p^\theta$, where $\mathsf{C}_1$ and $\mathsf{C}_2$ are constants. Further, the function $\Upsilon_{m,d}(\mathsf{k})$  can be written as [cf.\ Eq.~\eqref{eq:Yextrap}]
\begin{equation} \label{eq:Up}
\Upsilon_{m,d}(\mathsf{k})=\frac{1}{D_1}
\left[\frac{\mathcal{J}\big( \tilde{\mathsf{k} }^{-4}\ell_{\tilde{\mathsf{k}}}^{4\theta-2}\big)}
{\mathcal{B}_0\,\tilde{\mathsf{k}}^{2\epsilon} \ell_{\tilde{\mathsf{k}}}^{\epsilon(1-2 \theta)}}
\right]^{\alpha_{\mathrm{L}}/(\nu_{\mathrm{L}2}\epsilon)},\quad \tilde{\mathsf{k}}\equiv\mathsf{k}/D_2.
\end{equation}

In the limit $p\to 0$ at fixed $k>0$,  the momentum $\tilde{\mathsf{k}} \to\infty$ and 
$\ell_{\tilde{\mathsf{k}} }$ varies as ${\tilde{k}}^{1/\theta}$ according to Eq.~\eqref{eq:lkas}. Therefore  Eq.~\eqref{eq:Up} reduces to
\begin{equation} \label{Up1}
\Upsilon_{m,d}(\mathsf{k})\mathop{\approx}\limits_{\tilde{\mathsf{k}}\to \infty} \frac{1}{D_1}
\left[\frac{\mathcal{J}\big(\mathsf{ \tilde{k} }^{-2/\theta} \big)}
{\mathcal{B}_0\,\mathsf{\tilde{k}}^{\epsilon/\theta}  }
\right]^{\alpha_{\mathrm{L}}/(\nu_{\mathrm{L}2}\epsilon)}.
\end{equation}
Substituting this into our result for the energy-density cumulant $\tilde{G}^{(0,2)}_{\mathrm{R}}(\bm{k},\bm{p})$ 
yields
\begin{eqnarray} \tilde{G}^{(0,2)}_{\mathrm{R}}(\bm{k},\bm{p}) &\mathop{\approx}\limits_{p\to 0 }& \frac{\mathsf{C}_1}{D_1}
\left[\frac{\mathcal{J}\big(p^2\,{ {k} }^{-2/\theta} D_2^{2/\theta}\big)}
{\mathcal{B}_0\, {{k}}^{\epsilon/\theta} D_2^{-\epsilon/\theta} }
\right]^{\alpha_{\mathrm{L}}/(\nu_{\mathrm{L}2}\epsilon)} +\mathsf{C}_2. \label{GE} 
\end{eqnarray}
As we know from Eq.~\eqref{eq:Taylor}, $\mathcal{J}(w)$ has a Taylor expansion at $w=0$. Using this we see that the right-hand side of (\ref{GE}) is analytic in $p^2$ at $p\to0$ and fixed $k$, and hence  complies with the expansion~\eqref{eq:pexp}.

In the limit $k \to 0$ at fixed $p>0$,  we have $\tilde{\mathsf{k}} \to 0$ and  $\ell_{ \mathsf{\tilde{k}} }\approx 1$ from Eq.~\eqref{eq:lkas}. This implies
\begin{equation}
 \label{GE1}
\tilde{G}^{(0,2)}_{\mathrm{R}}(\bm{k},\bm{p})  \mathop{\approx}\limits_{k\to 0} \frac{\mathsf{C}_1}{D_1}
\left[\frac{\mathcal{J}\big(\tilde{\mathsf{k} }^{-4} \big)}
{\mathcal{B}_0\,p^{\epsilon}\,\tilde{\mathsf{k}}^{2\epsilon}  }
\right]^{\alpha_{\mathrm{L}}/(\nu_{\mathrm{L}2}\epsilon)} +\mathsf{C}_2.
\end{equation}
We can now insert the large-$w$ expansion~\eqref{eq:Jcallargewexp} for $\mathcal{J}(w)$ to conclude that our result for $\tilde{G}^{(0,2)}_{\mathrm{R}}(\bm{k},\bm{p})$  has indeed an expansion of the form~\eqref{eq:kexp}.

 \section{Summary and discussion} 
   \label{sec:concl}

In this paper we reconsidered Henkel's LSI theory for type-I systems and performed careful checks of its predictions. A major motive for our work was the apparently very good agreement of the LSI predictions with Monte Carlo results reported in Ref.~\cite{PH01}. Our paper has two qualitatively distinct parts. The first part dealt with the consequences of the conjectured invariance equations~\eqref{eq:X0def}--\eqref{eq:X1G}. Accepting these equations as given, we reanalyzed their solutions. As we have shown in Sec.~\ref{sec:rv} and Appendix~\ref{app:Om} for the cases $N=4$ and $N>4$, these equations generally have less physically acceptable solutions than anticipated by Henkel \cite{Hen99,Hen97,Hen02}.  Specifically, in the case $N=4$ that concerns  the comparison with Monte Carlo simulation for the three-dimensional ANNNI model \cite{PH01}, the contribution from the second linearly independent function,  $\Omega^{(4)}_1$, must be discarded for the following reasons.  If the scaling variable is taken to  involve the coordinate difference $t_1-t_2$ rather than its absolute value,  $\Omega^{(4)}_1(v)$ diverges exponentially  as $v\to-\infty$ and hence is unacceptable. We therefore made the replacement  $v\to |v|$ considering the function $\Omega^{(4)}_1(|v|)$. Rather than being a solution to Henkel's original homogeneous equation, our Eq.~\eqref{eq:YG}, this function turned out to be a solution to the inhomogeneous Eq.~\eqref{eq:y-1}, involving an inhomogeneity proportional to the derivative of the $\delta$ distribution. More importantly, we found that the contribution $\propto b_1^{(4)}$ entails a violation of general analyticity requirements (as discussed at the end of Sec.~\ref{sec:rv}).  Hence it is unacceptable. Its omission [by setting the coefficient $b^{(4)}_1=0$ in Eq.~\eqref{eq:Solg}], on the other hand, implies that the scaling function of the momentum-space energy-density cumulant of the LSI theory reduces to that of a free theory with the action~\eqref{eq:Hfree}, namely the function $\Upsilon_{m,d}^{\text{free}}$ given in Eq.~\eqref{eq:YLSI}.

In the second part of the paper we used RG-improved perturbation theory in $4+m/2-\epsilon$ dimensions to determine scaling functions of the order-parameter and energy-density cumulants in momentum space to two-loop order. The results are given in Sec.~\ref{sec:2lres}. For the special choice $d=m+3$, closed analytical expressions could be obtained for the scaling functions' series expansions  to second order in $\epsilon$. For the case of primary interest, the uniaxial case $m=1$, we managed to derive a countour-integral representation for the two-loop term of the momentum-space energy-density cumulant.

We found that the predictions of the  LSI theory generally are inconsistent with RG-improved perturbation theory in $d=d^*(m)-\epsilon$ dimension.  Only at the level of Landau theory for the order-parameter cumulant and the one-loop approximation for the energy-density cumulant, where the LSI theory yields scaling functions of massless effective free-field theories, did we find it to be in conformity with our systematic expansions for proper choices of the exponents $\tilde{\Delta}$ and $\theta$. However, as soon as we went beyond these orders to include nontrivial corrections to the scaling functions, the results did not comply with the LSI theory. In the uniaxial scalar $(m=n=1)$ 
case of the momentum-space energy-density cumulant we investigated in great detail, our $\epsilon$-expansion results for the scaling function turned out to be inconsistent with the LSI predictions irrespective of whether a contribution from the function $\Omega^{(4)}_1$ is taken into account $(b_1^{(4)}\ne0)$ or not $(b_1^{(4)}=0)$.

There are other observations concerning our two-loop results, which complement the evidence against the viability of the LSI theory provided by our $\epsilon$-expansion results. In Appendix~\ref{ap:mon}, we studied  the behavior of the function $\mathcal{J}(w)$ that the two-loop contribution to the energy-density cumulant involves in the complex $w$-plane. Unlike the LSI scaling function~\eqref{eq:Solg}, which has a single nontrivial singular point (located at $w=-\alpha/4$), the function $\mathcal{J}(w)$ was found to have additional singularities, namely branching points. The interested reader may find a detailed explanation of the branching behavior of this and the related functions $\mathcal{J}_1(w)$ and $\mathcal{J}_3(w)$ in that appendix. Their behavior in the complex $w$-plane differs qualitatively from that of the LSI function. We admit that these functions merely appear in RG-improved perturbation theory. However, given their qualitatively different behavior in the complex plane, it seems highly unlikely to us that proper resummations of the perturbation series might yield results in conformity with the LSI predictions, even if we did not know about the incompatibility with the $\epsilon$-expansion results. 

Finally, let us comment on the apparently excellent agreement of the LSI predictions with Monte Carlo simulation for the ANNNI model found in Ref.~\cite{PH01}. There are two problems with the LSI predictions used in the comparison: (i) they were based on the value $\theta=1/2$, which may be a good approximation but  differs from the RG estimate $\theta\simeq 0.47$ of Ref.~\cite{SD01} [and pretends that all corrections  of order $\Or(\epsilon^2)$ and higher to this classical value sum to zero when $\epsilon=3/2$]; (ii) a contribution proportional to the second linearly independent function, which we found to be problematic,  was taken into account. Using a value of  $\theta$ different from, but close to, $1/2$ in the scaling plot of the Monte Carlo results and dropping the non-free-field contribution to the LSI prediction will make the agreement presumably less striking, though it is not unlikely to remain reasonably good. If so, the situation would be reminiscent of the relatively good quality of the Ornstein-Zernike approximation  for the order-parameter two-point cumulant at a bulk critical point in $d=3$ dimensions whose reasons are twofold: the values of the correlation exponents $\eta(d=3,n)$ are close to the classical one $\eta_{\text{MF}}=0$ and corrections to the zero-loop result for the scaling function are small. One important ingredient for the eventual good agreement of the LSI prediction with the Monte Carlo data is the small deviation of $\theta$ from its classical value $1/2$.  It is conceivable that corrections to order-parameter scaling functions of the free LP theory are also small.  In fact, our extrapolation for the scaling function $\Upsilon_{1,3}$ of the energy-density cumulant presented in Fig.~\ref{fig:Yextrap} exhibits small deviations from the free-field theory LSI prediction $\Upsilon_{1,3}^{\text{free}}$, which in turn agrees with the one-loop approximation for the scaling function $\Upsilon_{1,3}$. 

In summary, we can draw two important conclusions: First, LSI theory is definitely not valid in a mathematical precise sense in the checked nontrivial case of critical behavior at Lifshitz points.  To our knowledge, the only cases in which its predictions are safely known to be exact are those  trivial ones in which it reproduces the results of free massless field theories. Hence, its predictive power and viability appears to be rather limited. Second, the seemingly very good agreement with Monte Carlo data reported in Ref.~\cite{PH01} is probably due to the fact that corrections to the Ornstein-Zernike theory happen to be small in this case of a uniaxial LP in $d=3$ dimensions. This need not be so  in other cases. The good agreement may therefore be deceptive. 

 \ack
We gratefully acknowledge partial support by Deutsche Forschungsgemeinschaft under Grant No.\ Di-378/3 at the early stages of this work. SR also gratefully acknowledges partial support by the Belorussian Republican Foundation for Fundamental Research.  One of us  (HWD) would like to thank Lothar Sch\"afer for a discussion and Helge Gr\"utjen for assistance in producing some of the figures. SR and MASh thank H.~W.\ Diehl and Fakult\"at f\"ur Physik for their hospitality at the Universit\"at Duisburg-Essen.
 \appendix
 \section{Calculation of the integrals $I_2(\mathsf{k})$ and $I_3(\mathsf{k})$}\label{app:I23}
 In this appendix we derive various results for the integrals $I_j(\mathsf{k})\equiv I_j(\mathsf{k};m,d)$ defined by Eq.~\eqref{eq:Jldefx}, including their Laurent expansions. Combining Eqs.~\eqref{eq:Jldefx} and \eqref{eq:Gfscf} yields
 \begin{equation}\label{eq:intrepIj}
I_j(\mathsf{k})=J_j(\bm{\mathsf{k}},\hat{\bm{p}};1)=\int d^mz\int d^{d-m}r\,r^{-j(2-\epsilon)}\,\Phi^j(zr^{-1/2})\,e^{i(\hat{\bm{p}}\cdot\bm{r}+\bm{\mathsf{k}}\cdot\bm{z})}.
\end{equation}
We first perform the angular integrations in the subspaces $\mathbb{R}^m$ and $\mathbb{R}^{d-m}$. Let
\begin{equation}
\label{eq:Angavdev}
X(m;\mathsf{k})\equiv \frac{1}{S_{m-1}}\int_{S^{m-1}}dA\, e^{i\bm{\mathsf{k}}\cdot\hat{\bm{z}}}
\end{equation} 
be the average of the function $\exp (i\bm{\mathsf{k}}\cdot \hat{\bm{z}})$ with $\hat{\bm{z}}\equiv\bm{z}/z$ over $S^{m-1}$. The function $X(m;\mathsf{k})$ may be found in the form of a Taylor series in Eq.~(A.17) of Ref.~\cite{DSZ03}. This series can be summed to obtain the closed-form expression
\begin{equation}
X(m;\mathsf{k})=\,_0F_1(m/2,-\mathsf{k}^2/4)=(\mathsf{k}/2)^{1-m/2}\,J_{m/2-1}(\mathsf{k})\,\Gamma(m/2),\;\;m>0.
\end{equation}

Using this result and making a change of variable $z\to v=zr^{1/2}$ in one of the radial integrations, the integrals $I_j$ can be written as
\begin{equation}
\label{eq:Ijexpr}
\frac{I_j(\mathsf{k};m,d)}{S_{m-1}\,S_{d-m-1}}
=\int_0^\infty dr\,r^{3-2j+(j-1)\epsilon}\,\varphi^{(j)}_{m,d}(r;\mathsf{k})
\end{equation}
with
\begin{equation}
\varphi^{(j)}_{m,d}(r;\mathsf{k})=\,_0F_1\bigg(\frac{d-m}{2},-\frac{r^2}{4}\bigg)\int_0^\infty dv\,v^{m-1}\,\Phi_{m,d}^j(v)\,\,_0F_1\bigg(\frac{m}{2},-\frac{r\mathsf{k}^2v^2}{4}\bigg),
\end{equation}
where $\epsilon=4+m/2-d$, as usual.

Let us first determine the values of $I_j$ at $\mathsf{k}=0$. When $\mathsf{k}=0$, the double integrals over $r$ and $v$ factorize. The $r$-integrals 
required  for $I_2(0;m,d)$ and $I_3(0;m,d)$ are analytically computable; the results are
\begin{equation}
 \label{eq:int0F1_2}
\int_0^\infty dr\,r^{\epsilon-1}\,_0F_1[(d-m)/2,-r^2/4]=\frac{2^{\epsilon -1} \Gamma(\epsilon/2)\, \Gamma(2-m/4-\epsilon/2)}{\Gamma(2-m/4-\epsilon)}
\end{equation}
and
\begin{equation}
 \label{eq:int0F1_3}
\int_0^\infty dr\,r^{2\epsilon-3}\,_0F_1[(d-m)/2,-r^2/4]=\frac{2^{2\epsilon -3} \,\Gamma(\epsilon-1)\, \Gamma(2-m/4-\epsilon/2)}{\Gamma(3-m/4-3\epsilon/2)},
\end{equation}
respectively.
The associated $v$-integrals are the special cases $J_{0,j}(m,d)$ of the integrals 
\begin{equation}
\label{eq:Jpjdef}
J_{l,j}(m,d)\equiv\int_0^\infty dv\,v^{m-1+l}\,\Phi_{m,d}^j(v),\;\;j=2,3;\;l=0,\ldots,4,
\end{equation}
previously used in Ref.~\cite{DS00a}. The first one, $J_{0,2}$, is known from Eq.~(82) of this reference for general values of $(m,d)$:
\begin{equation}
 \label{eq:JDS02}
J_{0,2}(m,d)=\frac{2^{-2-\epsilon}\,\Gamma^2(1-\epsilon/2)\,\Gamma(2-m/4-\epsilon)\,\Gamma(m/4)}{(2\pi)^d\,\Gamma(2-\epsilon)}.
\end{equation}
The result can be substituted along with Eqs.~\eqref{eq:int0F1_2} and \eqref{eq:Smm1} into expression~\eqref{eq:Ijexpr} for $I_2(\mathsf{k};m,d)$ to obtain
\begin{equation}
 \label{eq:I2zero}
I_2(0;m,d)=\frac{F_{m,\epsilon}}{\epsilon},
\end{equation}
in conformity with Eq.~(38) of Ref.~\cite{SD01}.

In a similar manner one finds 
\begin{equation}
\label{eq:I3zero}
I_3(0;m,d)=J_{0,3}(m,d)\,\frac{2^{7+m-2 d}  \,\pi^{d/2}\,\Gamma(3-d+m/2)}{\Gamma(3d/2-m-3)\,\Gamma(m/2)}.
\end{equation}
Unlike $J_{0,2}(m,d)$, the integral $J_{0,3}(m,d)$ is not known in closed analytical form for general values of $(m,d)$. However, it can be computed for given values of $(m,d)$ by numerical means \cite{DS00a,SD01}.

We now turn to the calculation of the Laurent expansion of the integrals $I_j$ in $\epsilon$ for general values of $m$. The right-hand side of Eq.~\eqref{eq:Ijexpr} may be read as the image  $\big(r_+^{(j-1)\epsilon+3-2j},\varphi^{(j)}_{m,d}\big)$ of the $r$-dependent test function $\varphi^{(j)}_{m,d}(r;\mathsf{k})$ under the map provided by the generalized function $r_+^{(j-1)\epsilon+3-2j}$ \cite{GS64}. These distributions have the Laurent expansions 
\begin{equation}
\label{eq:Laurdist}
r_+^{(j-1)\epsilon+3-2j}=\frac{1}{\epsilon}\,\frac{\delta^{(2j-4)}(r)}{(j-1)(2j-4)!}+r_+^{3-2j}+\epsilon\, r_+^{3-2j}\ln {r_+}+\Or(\epsilon^2), 
\end{equation}
where $\delta^{(l)}(r)$ is the $l$-th derivative of the $\delta$-distribution and the other distributions are defined by (see e.g.~Refs.~\cite{GS64} and \cite[Appendix]{Die86a})
\begin{eqnarray}
 \label{eq:rplusjdef}
(r_+^{-s}\ln^lr_+,\varphi)&=\int_0^\infty dr \,r^{-s}(\ln r)^l\bigg[ &\varphi(0)-\sum_{q=0}^{s-2}\frac{r^q}{q!}\frac{d^q\varphi}{d r^q}(0)\nonumber\\ &&\strut-\theta(1-r)\frac{r^{s-1}}{(s-1)!}\,\frac{d^{s-1}\varphi}{d r^{s-1}}(0)\bigg],
\end{eqnarray}
where $\theta(x)$  is the Heaviside step function.

The action of  $\delta^{(2j-4)}$ on $\varphi^{(j)}_{m,d}(r;\mathsf{k})$ can be computed in a straightforward manner. One obtains
\begin{equation}
 \label{eq:delvarphi2}
\big(\delta,\varphi^{(2)}_{m,d}(;\mathsf{k})\big)=\varphi_{m,d}^{(2)}(0;\mathsf{k})=J_{0,2}(m,d)
\end{equation}
and
\begin{equation}
 \label{eq:del2varphi3}
\big(\delta'',\varphi^{(3)}_{m,d}(;\mathsf{k})\big)=\varphi_{m,d}^{(3)\prime\prime}(0;\mathsf{k})=\frac{J_{4,3}(m,d)}{4m(m+2)}\,\mathsf{k}^4-\frac{J_{0,3}(m,d)}{d-m}.
\end{equation}

Combining Eqs.~\eqref{eq:Ijexpr}, \eqref{eq:int0F1_2}, \eqref{eq:Laurdist}, and \eqref{eq:delvarphi2} then gives
\begin{equation}\label{eq:I2Fexp}
\frac{I_2(\mathsf{k};m,d)}{F_{m,\epsilon}}=\frac{C_{-1}(m,\epsilon)}{\epsilon}+
\frac{S_{m-1}S_{d^*-m-1}}{F_{m,0}}\,\big(r_+^{-1},\varphi^{(2)}_{m,d^*}(;\mathsf{k})\big)
+\Or(\epsilon)
\end{equation}
with
\begin{eqnarray}
C_{-1}(m,\epsilon)&=&\frac{2^{-\epsilon}}{\Gamma(1+\epsilon/2)}
\frac{\Gamma(2-m/4-\epsilon)}{\Gamma(2-m/4-\epsilon/2)}
\\ &=&1+\frac{1}{2}\left[\gamma_E-2\ln 2-\psi(2-m/4)\right]\epsilon+\Or(\epsilon^2).
\end{eqnarray}
In order to be consistent with Eq.~\eqref{eq:I2zero}, the $\Or(\epsilon^0)$ term of the first expression  on the right-hand side of Eq.~\eqref{eq:I2Fexp} must cancel the $\mathsf{k}$-independent contribution of the second one. One easily checks that this is indeed the case and yields Eq.~\eqref{eq:R20m2} for $R^{(2)} _ 0 (\mathsf{k}; 2)$, a function that vanishes at $\mathsf{k}=0$.

Since the two-loop contribution to $G^{(2,0)}$ is quadratic in $\mathring{u}$, we split off a factor $F_{m,\epsilon}^2$ in $I_3$. We thus arrive at the Laurent expansion
\begin{eqnarray}\label{eq:I3Laurexp}
\frac{I_3(\mathsf{k};m,d)}{F_{m,\epsilon}^2}&=&\left[\frac{j_\sigma(m)\,\mathsf{k}^4}{16m (m+2)}-\frac{j_\phi(m)}{2(8-m)}\right]\bigg\{\frac{1}{\epsilon} +\gamma_E-2\nonumber\\ && \strut +\frac{\psi(2-m/4)-\ln(16\pi^3)}{2}\bigg\}+\big(r_+^{-3},\varphi^{(3)}_{m,d^*}(;\mathsf{k})\big)+\Or(\epsilon).
\end{eqnarray}

The terms in Eqs.~\eqref{eq:I2Fexp} and \eqref{eq:I3Laurexp} involving $r_+^{-1}$ and $r_+^{-3}$, respectively, cannot straightforwardly be evaluated in closed form for general values of $m$. However, we know from Eq.~\eqref{eq:frm2} that  the scaling functions $\Phi_{m,d}$ reduce to simple Gaussians on the line $d=m+3$. This enables us to determine the functions $\varphi^{(j)}_{m,d^*}$  in closed form for $m=2$. One obtains
\begin{equation}
 \label{eq:varphi225}
\varphi^{(2)}_{2,5}(r;\mathsf{k})=\frac{\sin r}{256\pi^4\,r}\,e^{-\mathsf{k}^2r/2}
\end{equation}
and
\begin{equation}
 \label{eq:varphi325}
\varphi^{(3)}_{2,5}(r;\mathsf{k})=\frac{\sin r}{6144\pi^6\,r}\,e^{-\mathsf{k}^2r/3}.
\end{equation}
To determine $\big(r_+^{3-2j},\varphi^{(j)}_{2,5}(;\mathsf{k})\big)$, we compute $(r_+^{3-2j+(j-1)\epsilon},\varphi^{(j)}_{2,5}(;\mathsf{k}))$, obtaining
 \begin{equation}
\int_0^\infty dr\, r^{\epsilon-1}\varphi^{(2)}_{2,5}(r;\mathsf{k})=-\frac{1}{256\pi^4}\Gamma(\epsilon-1)\,\Im[(\mathsf{k}^2/2+i)^{1-\epsilon}]
\end{equation}
and
 \begin{equation}
\int_0^\infty dr\, r^{2\epsilon-3}\varphi^{(3)}_{2,5}(r;\mathsf{k})=-\frac{1}{6144 \pi^6}\Gamma(2\epsilon-3)\,\Im[(\mathsf{k}^2/3+i)^{1-\epsilon}].
\end{equation}
The results can be Laurent expanded in $\epsilon$. The expansion coefficient of the $\epsilon^0$ terms then give us the required quantities $(r_+^{3-2j},\varphi^{(j)}_{2,5}(;\mathsf{k}))$. This leads to the results for $R^{(2)}_0(\mathsf{k};2)$ and $R^{(3)}_0(\mathsf{k};2)$ given in Eqs.~\eqref{eq:R20m2} and \eqref{eq:R30m2}, respectively.

To compute $R^{(2)}_1(\mathsf{k};2)$, we subtract from $I_2(\mathsf{k};m,d)$ its value at $\mathsf{k}=0$, defining
\begin{equation}
\hat{I}_2(\mathsf{k};m,d)\equiv I_2(\mathsf{k};m,d)-I_2(0;m,d).
\end{equation}
The subtraction eliminates the pole term. The desired function follows from the Taylor expansion of  $\hat{I}_2(\mathsf{k};m,d)$  to $\Or(\epsilon)$; we have
\begin{equation}
R^{(2)}_1(\mathsf{k};m)=\frac{\hat{I}_2^{(1)}(\mathsf{k};m)}{F_{m,0}}-
\big[1-\gamma_E/2+\ln(2\sqrt\pi)\big]R^{(2)}_0(\mathsf k;m),
\end{equation}
where $R^{(2)}_0(\mathsf{k};m)=\hat{I}_2(\mathsf{k};m,d^*)/F_{m,0}$
and $\hat{I}_2^{(1)}(\mathsf{k};m)\equiv\partial_\epsilon
\hat{I}_2(\mathsf{k};m,d^*-\epsilon)\big|_{\epsilon=0}$
is the coefficient of the $\Or(\epsilon)$ term  of $\hat{I}_2$,
while the coefficient in square brackets results from the $O(\epsilon)$
term of $F_{m,\epsilon}$.

We now specialize to the case $m=2$ and start from
\begin{equation}\label{eq:MITKA}
\hat{I}_2(\mathsf{k};2,5-\epsilon)=(2\pi)^{(3-\epsilon)/2}
\int_0^\infty dr\,r^{3(\epsilon-1)/2}J_{\frac{1-\epsilon}{2}}(r)
\int d^2 v\,\Phi_{2,5-\epsilon}^2(v)
\big(\e^{i\bm{\mathsf{k}}\cdot\bm v\sqrt{r}}-1\big).
\end{equation}
For the Bessel function we substitute its $\epsilon$~expansion.
It is convenient to use the integral representation of the $\Or(\epsilon)$
term's coefficient $-\partial_\nu J_\nu(\left.r)\right|_{\nu=1/2}/2$ one  obtains from
(see e.g. \cite[entry 2.3.5.3]{PBM1})
\begin{equation}
J_\nu(r)=\frac{(r/2)^\nu}{\Gamma(\nu+1/2)\,\sqrt\pi}
\int_{-1}^1 dt\,e^{i r t}(1-t^2)^{\nu-1/2}\,,\quad\quad\Re\,\nu>-1/2,
\end{equation}
by differentiating both sides and interchanging the integration and differentiation on the right-hand side. This gives
\begin{eqnarray}\label{eq:JEX}
J_{\frac{1-\epsilon}{2}}(r)&=&
\sqrt{\frac{2}{\pi r}}\bigg[1-\frac{\epsilon}{2}\Big(\gamma_E+\ln\frac{r}{2}\Big)\bigg]
\sin r\nonumber\\&&\strut
-\epsilon\,\sqrt{\frac{r}{2\pi}}
\int_0^1 dz\,\cos (rz) \ln(1-z^2)+\Or(\epsilon^2).
\end{eqnarray}
The last integral can be done analytically yielding the known result in terms of
sine and cosine integrals
\cite[p. 74]{MOS66} quoted in Eq.~(49) of \cite{ApelKrav85}.

The square of the scaling function $\Phi_{2,5-\epsilon}(v)$ is treated in an analogous fashion. We replace $\Phi_{2,5-\epsilon}(v)$ by its expansion to $\Or(\epsilon)$ and use an appropriate integral representation for the term linear in $\epsilon$. A convenient starting point to find the latter is the integral represention of the scaling function
\begin{equation}\label{eq:intrepPhi}
\Phi_{m,4+m/2-\epsilon}(v)=\frac{2^{-m-1}\pi^{(2\epsilon-6-m)/4}}{\Gamma[(m-2+2\epsilon)/4]}\int_0^1dt\,
t^{1-\epsilon}(1-t^2)^{(m-6+2\epsilon)/4}\,e^{-t v^2/4}.
\end{equation}
One way to prove it is to perform the integration with the aid of {\sc Mathematica} \cite{Mathematica7} to obtain the explicit result~\eqref{eq:Phi} when $\Re(m + 2 \epsilon) > 2$ and  $\Re\,\epsilon < 2$ (which can then be analytically continued). Alternatively, one can Taylor expand the exponential in Eq.~\eqref{eq:intrepPhi} and integrate term by term. The result is the Taylor series of $\Phi_{m,d}$ given by Eqs.~(10) and (11) of Ref.~\cite{SD01}.

Setting $m=2$ and computing the Laurent expansion of  the integral in Eq.~\eqref{eq:intrepPhi} to order $\epsilon^0$, one can show that the scaling function satisfies
\begin{equation}\label{eq:phim2Laurexp}
\Phi_{2,5-\epsilon}(v)=\frac{\pi^{-2+\epsilon/2}}{16\,\Gamma(1+\epsilon/2)}\left[e^{-v^2/4}-\epsilon\,\frac{ v^2}{8}\int_0^1 dt\,e^{-tv^2/4}\ln(1-t^2)+\Or(\epsilon^2)\right].
\end{equation}
To prove this, we rewrite the power $(1-t^2)^{-1+\epsilon/2}$ of the integral's integrand as $-\partial_t (1-t^2)^{\epsilon/2}/\epsilon$ and integrate by parts. The contribution from the boundary term can be rewritten as the limit $a\to 0+$ of
\[\frac{a^{-\epsilon}}{\epsilon}=-\frac{1}{\epsilon}-\int_a^1dt\,t^{-1-\epsilon}.
\]
The integral on the right-hand side can be combined with one of the two integrals produced via integration by parts to obtain
\begin{equation}\label{eq:A30}
-\int_0^1dt\, t^{-1-\epsilon}[(1-t^2)^{\epsilon/2}e^{-tv^2/4}-1]=
\gamma_E+\text{E}_1(v^2/4)-\ln(v^2/4)+\Or(\epsilon).
\end{equation}
In the remaining integral  $(-v^2/4\epsilon)\int_0^1dt\,e^{-tv^2/4}(1-t^2)^{\epsilon/2}\,t^{-\epsilon}$ we expand the $\epsilon$-dependent powers to $\Or(\epsilon)$ and perform the two integrals that do not involve $\ln(1-t^2)$. Adding up all  contributions and multiplying by the prefactor then gives the stated result~\eqref{eq:phim2Laurexp}. The integral remaining in Eq.~\eqref{eq:phim2Laurexp} is analytically calculable and can be expressed in terms of
the exponential integral functions $\text{Ei}$ and $\text{E}_1$ \cite{AS72}.
However, we found it more convenient to work with the integral
representation~\eqref{eq:phim2Laurexp} for $\Phi_{2,5-\epsilon}$, rather than with the analytic expressions in terms of special functions. Likewise, we prefer to use the integral representation  of $J_{(1-\epsilon)/2}$ given in the first two lines of Eqs.~\eqref{eq:JEX}.

Upon substituting them into Eq.~\eqref{eq:MITKA}, the required Gaussian integrations over $v$ can be done in a straightforward fashion, giving
\begin{eqnarray}\nonumber\label{eq:TRE}
\frac{\hat I_2(\mathsf{k};2,5-\epsilon)}{F_{2,0}}&=&Q_0(\mathsf{k}^2/2){\left[1+\frac{\epsilon}{2}\,(\gamma_E+\ln\pi)\right]}
\nonumber\\ &&\strut +
\epsilon\int_0^\infty dr\, \frac{\sin r}{r^2}
\Big(\e^{-r\mathsf k^2/2}-1\Big)\ln r
\nonumber\\&&\strut
-\frac{\epsilon}{2}\,\int_0^1 dz \ln(1-z^2)
\int_0^\infty dr \,\frac{\cos (rz)}{r}\big(\e^{-r\mathsf k^2/2}-1\big)
\nonumber \\ &&\strut
+2\epsilon \int_0^1 dz\, \frac{\ln(1{-}z^2)}{(1+z)^2}
\int_0^\infty dr\, \frac{\sin r}{r^2}
\bigg[1
\nonumber\\ &&\strut\qquad\qquad\qquad\quad
 -\Big(1-\frac{r\,\mathsf{k}^2}{1+z}\Big)
\exp\left(-\frac{r\,\mathsf{k}^2}{1+z}\right)\bigg]
+O(\epsilon^2)
\end{eqnarray}
where $Q_0(s)$ is the function defined in Eq.~\eqref{eq:Q0def}.
Performing the remaining $r$ integrations, one arrives at
\begin{equation}\label{eq:FOR}
R^{(2)}_1(\mathsf k;2)=Q_1(\mathsf{k}^2/2)+\frac{1}{4}W_1(\mathsf{k}^2/2)+
W_2(\mathsf{k}^2/2)-Q_0(\mathsf{k}^2/2)\ln 2,
\end{equation}
where $Q_1(s)$ is the function~\eqref{eq:Q1def}, while
$W_1(s)$ and $W_2(s)$ denote the integrals
\begin{equation}
W_1(s)=\int_0^1 dz\, \ln(1-z^2)\ln\Big(1+\frac{s^2}{z^2}\Big)
\end{equation}
and
\begin{equation}
W_2(s)=\int_0^1 dz \,\frac{\ln(1-z^2)}{(1+z)^2}
{\left\{4\frac{2s}{1+z}\arctan{\Big(\frac{1+z}{2s}\Big)}
+\ln{\bigg[1+\frac{4s^2}{(1+z)^2}\bigg]}\right\}}.
\end{equation}

To compute $W_1$, we express $\ln(1-z^2)$ as a  sum of logarithms  $\ln(1\pm z)$. The integral can then be evaluated using {\sc Mathematica} \cite{Mathematica7}. The result involves the dilogarithm function $\mbox{Li}_2[2i/(i+s)]$. The real part of this expression can be written as $\mbox{Li}_2(2\cos\varphi,\varphi)$  in the notation of Ref.~\cite{Lewin}, with $\varphi=\arg[2i/(i+s)]=\arctan s$.
According to its Eq.~(5.17), it is given by
$\mbox{Li}_2(2\cos\varphi,\varphi)=\arctan^2(1/s)\equiv\vartheta^2(s)$
where $\vartheta(s)$ is the function introduced in Eq.~\eqref{eq:defvartheta}.
To rewrite the imaginary part of the expression, we use the
inversion formula for the dilogarithm (see e.g.\ Refs.~\cite[Eq.~(1.10)]{Lewin} and
\cite[p.~652]{PBM3}),
\begin{equation}
\mbox{Li}_2(z)=\frac{\pi^2}{3}-\mbox{Li}_2(1/z)-\frac{1}{2}\ln^2 z
+\pi i\ln z\,, \quad\quad\quad |\arg(-z)|<\pi\,.
\end{equation}
This gives
\begin{eqnarray}
\Im{\bigg[\mbox{Li}_2\Big(\frac{2i}{i+s}\Big)\bigg]}&=&
\frac{1}{2}\ln{\Big(\frac{1+s^2}{4}\Big)}[\arctan(s)-\pi]
+\Im{\bigg[\mbox{Li}_2\Big(\frac{1}{2}+i\frac{s}{2}\Big)\bigg]}\nonumber\\
&=&
-\vartheta(s)\ln\frac{1+s^2}{4}+\mbox{Cl}_2(2\arctan s)
-\frac{1}{2}\mbox{Cl}_2(4\arctan s),\nonumber\\
\end{eqnarray}
where Eq.~(5.5) of Ref.~\cite{Lewin} was used to arrive at the second expression.
Application of the duplication formula \cite[Eq.~(4.17)]{Lewin} to the two Clausen
integrals $\mbox{Cl}_2(\theta)$ along with the relation
$\pi/2-\arctan s=\arctan(1/s)$ finally yields
\begin{equation}\label{eq:WOD}
W_1(s)=\frac{\pi^2}{2}+4(1-\ln 2)\,Q_0(s)-2 \vartheta^2(s)
 +2 s\, \vartheta(s)\ln[(1+s^2)/4]-
2 s\, \mbox{Cl}_2[2\vartheta(s)].
\end{equation}

The calculation of the integral $W_2$  proceeds
along similar lines but is more cumbersome and lengthier.
Without entering into details, we just record our final result:
\begin{eqnarray}\label{WDVA}%\nonumber
W_2(s)&=&-\frac{\pi^2}{12}-\frac{1}{2}-\ln\sqrt 2\,\ln\frac{1+s^2}{2}+\ln\sqrt{1+4 s^2}
+\Big(\frac{1}{2}-\ln 4\Big)\,s\,\vartheta(s)\\%\nonumber
&&\strut +\left[\vartheta(s)-\vartheta(2s)\right]\Big[\frac{1}{2s}-
s\,\ln\frac{\sqrt{1+s^2}}{s}\Big]+
\Re\left[\mbox{Li}_2\Big(\frac{i}{2i+2s}\Big)\right]+\frac{1}{4}\mbox{Li}_2(-4s^2)
\nonumber\\&&\strut
-\frac{1}{4}\mbox{Li}_2(-s^2)\nonumber%\\&&\strut
+\frac{s}{2}\left\{\mbox{Cl}_2[2\vartheta(2s)]-\mbox{Cl}_2[2\vartheta(s)]
+\mbox{Cl}_2[2\vartheta(s)-2\vartheta(2s)]\right\}.
\end{eqnarray}
From the above equations, the result for $R^{(2)}_2(\mathsf{k};2)$ given in Eq.~\eqref{eq:R21m2} follows in a straightforward manner.

It is also possible to compute the integrals $I_2$ and $I_3$ on the line $d=m+3$ for general values of $m$ with $0\le m\le 2$. One finds
\begin{equation}
 \label{eq:I2deqmp3}
\frac{I_2(\mathsf{k};m,m+3)}{F_{m,1-m/2}}=\frac{1}{1-m/2}\,\frac{\left(\mathsf{k}^2-2i\right)^{m/2}-\left(\mathsf{k}^2+2i\right)^{m/2}}{(1-i)^m-(1+i)^m}
\end{equation}
and
\begin{equation}
 \label{eq:I3deqmp3}
I_3(\mathsf{k};m,m+3)=\frac{\Gamma(-m-1)}{i 2^{5+2m}\,3^{1+3m/2}\pi^{2+m}}\,\left[\left(\mathsf{k}^2-3i\right)^{m+1}-\left(\mathsf{k}^2+3i\right)^{m+1}\right].
\end{equation}
Using these results and setting $m=d-3=2-2\epsilon$, one can determine the scaling functions $\Psi_{2-2\epsilon,5-2\epsilon}$ and $\Upsilon_{2-2\epsilon,5-2\epsilon}$ to $\Or(\epsilon^2)$ in a straightforward manner. The results are given in Eqs.~\eqref{eq:Psidm3exp}--\eqref{eq:Q1def}.

 \section{Integral representations for the scaling function $\Omega(v)$\label{app:Om}}
The differential equation (\ref{eq:Hv}) plays a central role in Henkel's theory. He expressed its  general solution (\ref{eq:hyp}) in terms of the generalized hypergeometric functions ${}_2F_{N-1}$. However, it is quite difficult to analyze the behavior of the function $\Omega^{(N)}(v)$  for $v\to\infty$ in  this representation. Here we derive an alternative integral representation for the general solution of Eq.~\eqref{eq:Hv}, which is more convenient for this purpose.

 The main result of this appendix is the following theorem.
 \begin{theorem}
  Let $N>2$ be an integer and
 $\zeta>1$. Then the general solution of the differential equation
 \begin{equation} \label{eq:Hu}
  L\,f(v)\equiv\left(\frac{d^{N-1}}{dv^{N-1}}-v^2 \frac{d}{dv}-\zeta\,v\right)f(v)=0,
 \end{equation}
 is given by a linear combination of the $(N-1)$ functions $f_l(v)$ defined by
 \begin{equation}
   \label{eq:fll}
  f_l(v)=\sigma_l\, g(\sigma_l \,v)-g(v)
  \end{equation}
  with 
  \begin{equation}\label{eq:gu}
 g(v)=\int_0^\infty dk\, k^{\nu N/2}\,{\exp}{\Big[(N/2)^{2/N} v k\Big]}
 \, K_\nu(k^{N/2}),\;\;
 l=1,\ldots,N-1,
 \end{equation}
 \begin{equation*}
 \sigma_{l}=\exp (2\pi i l/N),
 \end{equation*}
 and
 \begin{equation*}
\nu=\frac{\zeta-1}{N}.
 \end{equation*}
 \end{theorem}
 
 \begin{proof}
Upon integrating by parts and using the Bessel differential equation for the Macdonald function, one can show that
  \begin{equation}
    \label{eq:Lg}
 L\,\big[\sigma_l \,g(\sigma_l\,v)\big]=(2/N)^{2/N-1}\,2^\nu \,\Gamma(\nu+1).
  \end{equation}
 This proves the theorem since the right-hand side in
 (\ref{eq:Lg}) is independent of $l$.
 \end{proof}

Several remarks are in order here. Note, first, that the  function $g(v)$ defined by the integral in Eq.~\eqref{eq:gu} is holomorphic in the complex $v$-plane. Second, if  $f(v)$ is a solution to Eq.~\eqref{eq:Hu}, then the function
 $\Omega(v)=f(v\, \alpha^{-1/N})$, with $\alpha\neq0$ an arbitrary complex number, solves Eq.~\eqref{eq:Hv}. Third, since the functions $f_{l}(v\, \alpha^{-1/N})$ and $v^p\,\mathcal{F}_p$ provide  two sets of $N-1$ linearly independent solutions of Eq.~\eqref{eq:Hv}, they must be related by a $v$-independent  matrix. This yields, on the one hand, integral representations for the generalized hypergeometric functions~\eqref{eq:Fp} and, on the other hand, expresses the integrals $f_{l}$ as linear combinations of the functions  $v^p\, {}_2F_{N-1}$. For given integer $N=3,4,\ldots$, the coefficients of these linear combinations can be determined explicitly by integrating Eq.~\eqref{eq:gu} using {\sc Mathematica} \cite{Mathematica7}. Taking into account that
 \begin{equation}
   \label{eq:Mc}
   K_\nu(z)\mathop{=}_{z\to\infty}\sqrt{\frac{\pi}{2 z}} \, e^{-z}[1+ O(1/z)]\quad\mbox{for }
 |\arg z|<3\pi/2, 
 \end{equation}
 one can find from (\ref{eq:gu}) the asymptotic behavior of the function $g(v)$ as $v\to\infty$ via the saddle point method. Its  limiting form depends on $\arg v$  and $N$. It is governed either by a saddle point of the integrand of the integral in Eq.~\eqref{eq:gu} or else by the integrand's behavior in the vicinity of the origin $k=0$. Since the complete analysis of the asymptotic large-$|v|$ behavior of the function $g(v)$ for generic $\arg v$ and $N$ 
is rather involved, we will focus our attention on those special cases that are relevant for the $v\to\infty$ asymptotics of the  function $\Omega(v)$.

Upon substituting the limiting large-$z$ form~\eqref{eq:Mc} for the Bessel function of the  integrand of the integral~\eqref{eq:gu}, the integrand becomes $F(k)\exp[\mathcal{R}(v,k)]$ with $F(k)=\sqrt{\pi/2}\,k^{N(2\nu-1)/4}$ and 
\begin{equation}
    \label{eq:Fu}
   {\mathcal R}(v,k) =\left(N/2\right)^{2/N}vk-k^{N/2}
  \end{equation}
for large positive $k$. For given $v>0$,  ${\mathcal R}(v,k)$ takes its maximum on the integration path $0<k<+\infty$ at the $k$-value
 \begin{equation}
   \label{eq:k0}
 k_s(v) =\left(2 /N\right)^{2/N}\,v^{2/(N-2)}.
 \end{equation}
 To obtain the asymptotic behavior of the function $g(v)$ as $v\to\infty$, we can therefore expand $\mathcal{R}(v,k)$ about $k_s(v)$ to second order, replace $F(k)$ by $F(k_s)$, and extend the lower and upper integration limits of the resulting Gaussian integral to $\pm\infty$. We thus arrive at the asymptotic behavior 
  \begin{eqnarray}
 \label{eq:asg}
   g(v)\mathop{=}_{v\to\infty}\pi \sqrt{\frac{N}{N-2}}\left( \frac{2}{N}\right)^{\frac{1+\zeta}{N}}
 v^{\frac{1-N+\zeta}{N-2}}  
   \exp{\bigg(\frac{N-2}{N}\,v^{\frac{N}{N-2}}\bigg)}{\left[1+ \Or(v^{-1})\right]}.
 \end{eqnarray}
 
To compare this with Henkel's results, we set $\alpha_1=1$ in equations (4.23) and  (4.24) of Ref.~\cite{Hen02}. The comparison with  Eq.~\eqref{eq:asg} shows that the leading exponential large-$v$ divergence given in Eq.~(4.24) of Ref.~\cite{Hen02} originates from $g(v)$. Henkel's condition (4.25) \cite{Hen02}, which we reproduce in Eq.~\eqref{eq:linb} simply means that the function $g(v)$ is required not to contribute to $\Omega(v)$ so that the leading contribution to $\Omega(v)$ as $v\to \infty$ results  from the term $\sigma_1 g(\sigma_1 v)$ in Eq.~\eqref{eq:fll}.

Consider next the large $v$-asymptotics   of $ g(\sigma_1 v)$ for integer values $N\ge 4$. To this end, we must study the behavior of the function~\eqref{eq:gu} as $v\to \infty$ at fixed $\arg  v=2\pi/N$. The asymptotic form can be derived along lines similar to those followed to obtain Eq.~\eqref{eq:asg}, provided the angle $\arg  v$ is small enough. 
For such $\arg  v>0$, the saddle point~\eqref{eq:k0} is located in the upper complex $k$-plane slightly above the real axis. The integration path must be deformed into the steepest-decent curve crossing the complex  saddle point $k_s(\sigma_1|v|)$. The resulting contribution may be gleaned from Eq.~\eqref{eq:asg} by substituting $v\to \sigma_1 v$. It reads
\begin{eqnarray}
   \label{eq:asg1}
   g(\sigma_1v) &\mathop{\approx}_{v\to \infty}&\pi \sqrt{\frac{N}{N-2}}\,\Big(\frac{2}{N}\Big)^{(1+\zeta)/N}\,
 (\sigma_1v)^{\frac{1-N+\zeta}{N-2}}\,\nonumber\\
 &&\times  \exp\left\{\frac{N-2}{N}v^{N/(N-2)}\,e^{i\, 2\pi/(N-2)}\right\}
\left(1+ O(v^{-1})\right). \nonumber\\
 \end{eqnarray}
and gives the asymptotic large-$v$ behavior  when $N\ge 6$. For $N=6$, the limiting form simplifies to \begin{eqnarray}
   \label{eq:asg3}
   g(\sigma_1v)\mathop{\approx}_{v\to\infty} \pi\sqrt{3/2}\,\, 3^{-(1+\zeta)/6} [v \exp(\pi i/3)]^{(\zeta-5)/4} \exp(2  i v^{3/2}/3).
\end{eqnarray}
Thus, $g(\sigma_1v)$ decays indeed as $v^{(\zeta-5)/4}$ in this case. By contrast, when $N\ge 7$, the right-hand side of Eq.~\eqref{eq:asg1} diverges exponentially as $v\to+\infty$. In order to prevent such subleading divergences in the scaling functions $\Omega(v)$, further constraints must be imposed in addition to Eq.~\eqref{eq:linb} on the  coefficients $b_p^{(N)}$ in Eq.~\eqref{eq:hyp}. Otherwise the algebraic decay assumed by Henkel does not apply.

Turning to the case $4\le N<6$, we note that $\cos [2 \pi /(N-2)]<0$ for $N<6$. Therefore, the right-hand side of Eq.~\eqref{eq:asg1} decays exponentially and does not describe the asymptotic behavior as $v\to\infty$.  A straightforward  analysis reveals  that the leading contribution to the integral~\eqref{eq:gu} results from vicinity of the end point $k=0$ of the integration path, giving
\begin{eqnarray}
   \label{eq:asg2}
   g(\sigma_1\,v)\mathop{\approx}_{v\to \infty}-2^{\nu N/2-1}\left(\frac{2}{N}\right)^{N/2}\frac{\Gamma(\nu)}{ \sigma_1\, v}+O(v^{-2})
 \end{eqnarray}
for  $4\le N< 6$. In this case,  all divergences of $\Omega(v)$ are contained in $g(v)$, which is in agreement with Henkel's numerical check \cite{Hen02}. 
 
 \section{Contour integral representation for one-loop function \label{app:Cont}}
 In this Appendix we show that the one-loop Feynman integral $J(P)$ introduced in Eq.~\eqref{eq:JPdef} has the contour-integral representation given in Eqs.~\eqref{eq:JcalJ}--\eqref{eq:Jcal3}. Choosing Cartesian coordinates $p_1,\ldots,p_{d-1}$ such that $\hat{\vec{P}}\equiv \vec{P}/P$ points along the $p_1$-axis, we decompose $\bm{p}$ as  $\bm{p}=p_1\hat{\vec{P}}+\bm{p}_\perp$. The resulting integrand of $\mathcal{J}(P)$ is a rational function of $p_1$ having four simple poles located at $\pm i \sqrt{\bm{p}_\perp^2+k^4}$ and $P\pm i \sqrt{\bm{p}_\perp^2+(k-1)^4}$. To perform the integration over $p_1$, we close the contour in the upper half plane. The result then becomes a sum of $2\pi i$ times the residues at the poles in the upper half plane $\Im\,p_1>0$. In the contribution from the residue with $\Re\ p_1=P$, we make the  change of variables $k-1\to k$. We thus see that $\mathcal{J}(P)$ can be written as
 \begin{equation}
 \label{eq:JcalI}
 \mathcal{J}(P^2)= \Re[\mathcal{I}(P^2)]
\end{equation}
with
\begin{eqnarray}
\mathcal{I}(P^2)&=&
 \frac{1}{(2\pi)^{d-1}}\int_{-\infty}^\infty d k \int d^{d-2} {\bm
 p}_\perp \frac{1}{({\bm p}_\perp^2+k^4)^{1/2}} \nonumber \\
 &&\times  \frac{1}{(k-1)^4-k^4+P^2-2 i P\, (\bm{p}_\perp^2+k^4)^{1/2}}.
 \end{eqnarray}
 
The next step is to split the  integration over $k$ into two parts, one from   $-\infty$ to $0$, and a second one from 
 $0$ to $\infty$. This leads to
 \begin{equation}
 \label{eq:calI}
 \mathcal{I}(P^2)=\mathcal{I}_+(P^2)+\mathcal{I}_- (P^2)
 \end{equation}
 with
 \begin{eqnarray}
  \mathcal{I}_\pm (P^2)&=&\int_{0}^\infty \frac{d k}{2\pi}
 \int \frac{d^{d-2}p_\perp}{(2\pi)^{d-2}}\,\frac{1}{({\bm p}_\perp^2+k^4)^{1/2}}\nonumber\\ &&\times 
  \frac{1}{(k\pm 1)^4-k^4+P^2-2 i P ({\bm
 p}_\perp^2+k^4)^{1/2}},
 \end{eqnarray}
 where the positive values of both square roots are chosen. In the integral over $\bm{p}_\perp$ we perform the angular integrations and make the change of variables $p_\perp\to\theta$ with $p_\perp=k^2\sinh\theta$. This gives
 \begin{equation}
  \label{eq:calIpm}
 \mathcal{I}_\pm (P^2)=\frac{S_{d-3}}{(2\pi)^{d-1}}\int_{0}^\infty d
 \theta\, \sinh^{d-3}\theta \int_0^\infty d k
 \frac{k^{2d-6}}{g_\pm(k,\theta)}
 \end{equation}
 with
 \begin{equation}
   \label{eq:gpm}
 g_\pm (k,\theta)=(k\pm1)^4-k^4+P^2-2 i P k^2\cosh \theta.
 \end{equation}
 
Let us first consider the integral ${\mathcal I}_- (P^2)$. To compute its inner integral in Eq.~\eqref{eq:calIpm} by means of residue calculus, we combine the integrals along paths infinitesimally above and below the real axis to conclude that
 \begin{equation}
    \label{eq:kintC}
 \int_0^\infty d k
 \frac{k^{2d-6}}{g_-(k,\theta)}=\frac{1}{1-e^{4\pi i d}}\int_{\mathcal{C}} d k
 \frac{k^{2d-6}}{g_-(k,\theta)},
 \end{equation}
 where $\mathcal{C}$ is the integration path shown in Fig.~\ref{fig:intPathC}. \begin{figure}[htbp]
   \centering
 \includegraphics[width=0.6\columnwidth]{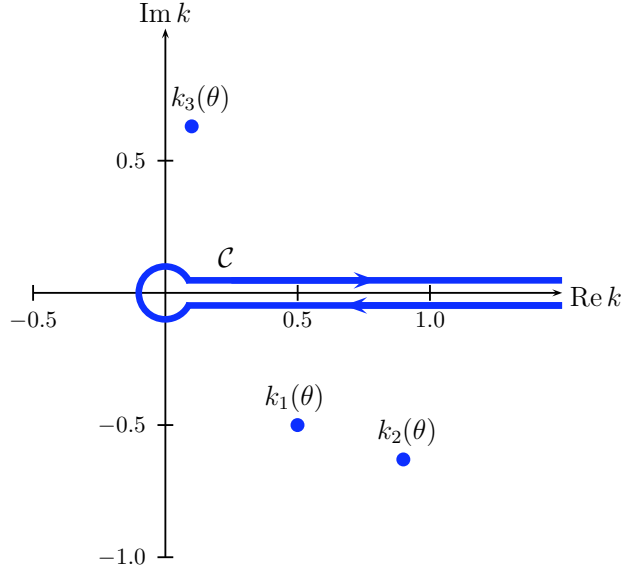}
   \caption{Integration path $\mathcal{C}$ of the integral on the right-hand side of Eq.~\eqref{eq:kintC} and poles $k_j(\theta)$ of the integrand. The integration contour $\mathcal{C}$ can be closed by a circle of radius $R\to\infty$.}
   \label{fig:intPathC}
 \end{figure}
The integrand has three simple poles away from the real axis, located at  the zeros $k_j(\theta)$ of the cubic equation
 \begin{eqnarray}
 \label{eq:gmmin}
\left.g_-(k,\theta)\right|_{k=k_j(\theta)} &=&\left[1-4 k^3-2 i k^2 P \cosh (\theta
   )+6 k^2-4 k+P^2\right]_{k=k_j(\theta)}\nonumber\\ &=&0,\quad j=1,2,3. 
 \end{eqnarray}
 We choose them in such a way that
 \begin{eqnarray}\label{eq:k1to3}
 &k_1(0)=\frac{1}{2}-\frac{i P}{2}, \quad &\pi<\arg k_{1}(\theta)<2\pi,\nonumber\\
 &k_2(0)=\frac{1}{2}+\frac{1}{2}(-1-2 i P)^{1/2},\quad&\pi<\arg k_{2}(\theta)<2\pi,  \\
 &k_3(0)=\frac{1}{2}-\frac{1}{2}(-1-2 i P)^{1/2} ,\quad&0<\arg k_{3}(\theta)<\pi,\nonumber
 \end{eqnarray}
 for $0\le\theta<\infty$ and real $P>0$. 
 
Since the integration contour $\mathcal{C}$ can be closed by a circle of radius $R\to\infty$, we can apply the residue theorem.  Upon exploiting  Eqs.~\eqref{eq:gpm} and \eqref{eq:gmmin}, we find for the residues
  \begin{eqnarray}
 \mathop{\text{Res}}_{k=k_j(\theta)}\bigg[\frac{1}{g_-(k,\theta)}\bigg]&=&\left[\frac{\partial g_-(k,\theta)}{\partial k}\right]_{k=k_j(\theta)}^{-1}=-k_j'(\theta)
 \left[\frac{\partial g_-(k,\theta)}{\partial
 \theta}\right]_{k=k_j(\theta)}^{-1}\nonumber\\ &=& \frac{k_j'(\theta)}{ 2 i P
 [k_j(\theta)]^2 \sinh\theta}.
 \end{eqnarray}
It follows that
 \begin{equation}
 \int_0^\infty d k
 \frac{k^{2d-6}}{g_-(k,\theta)}=\frac{\pi}{(1-e^{4\pi i d})\, P }\sum_{j=1}^3
 \frac{[k_j(\theta)]^{2d-8}}{\sinh \theta}\,k_j'(\theta)\,,
\end{equation}
which inserted into Eq.~\eqref{eq:calIpm}, then yields
 \begin{equation} \label{eq:I-}
 \mathcal{I}_-(P^2)=\frac{S_{d-3}}{2\,P\,(2\pi)^{d-2}}\,\frac{1}{1-e^{4\pi i d
 }}  \sum_{j=1}^3 \int_0^\infty d\theta \, 
 k'_j(\theta)\,\big[k_j^2(\theta)\sinh \theta\big]^{d-4}.
 \end{equation}

 The integral ${\mathcal I}_+(P)$ can be dealt with in a similar manner. The poles of the $k$-integral are now given by the zeros $k_j(\theta)$, $j=4,5,6$, of the function $g_+(k,\theta)$, namely
  \begin{equation}
   \label{eq:k4to6}
 k_{4}(\theta)=k_{1}(\theta)\,e^{-i\pi},\quad k_{5}(\theta)=k_{2}(\theta)\,e^{-i\pi},\quad k_{6}(\theta)=k_{3}(\theta)\,e^{i\pi}.
 \end{equation}
Note that the chosen phases in Eqs.~\eqref{eq:k1to3} and \eqref{eq:k4to6} guarantee that  $0<\arg k_j(\theta)<2\pi$ for all $j=1,...,6$ when $0<\theta<\infty$ and $P>0$. The analog of Eq.~\eqref{eq:I-} becomes
 \begin{equation}\label{eq:I+}
 {\mathcal I}_+(P^2)=-\frac{S_{d-3}}{2\,P\,(2\pi)^{d-2}}\,
 \frac{e^{-2 i \pi d}}{1-e^{4\pi i d }}  \sum_{j=1}^3 e^{2i\pi d\,\delta_{j3}} \int_0^\infty d\theta \,
 k'_j(\theta)\big[k_j^2(\theta)\sinh \theta\big]^{d-4}.
 \end{equation}

 It is convenient to express the integrals over $\theta$ on the right-hand side of this equation in terms of the complex integration variables $\varsigma$:
 \begin{equation}\label{eq:defk}
 \varsigma=\varphi_j(\theta)\equiv k_j(\theta)-1/2,\;\;j=1,2,3.
 \end{equation}
 These maps $\varphi_j:(0,\mathbb{R})\to \mathbb{C}$ parametrize paths $\mathcal{C}_j$ in the complex $k$-plane. The total integral~\eqref{eq:calI} therefore becomes an integral along the curve $\mathcal{C}_3-\mathcal{C}_1-\mathcal{C}_2$,
 \begin{equation}\label{eq:IP}
 \mathcal{I}(P^2)=\frac{S_{d-3}}{4\,P\,(2\pi)^{d-2}\cos (\pi d) } \,
  \int_{\mathcal{C}_3-\mathcal{C}_1-\mathcal{C}_2} d \varsigma \, [p_\perp (\varsigma)]^{d-4},
 \end{equation}
 where
 \begin{equation}\label{eq:pperp}
 p_\perp (\varsigma)=\sqrt{-\frac{4 \varsigma^2+P^2}{16P^2}\,(4P^2+1+8 \varsigma^2+16 \varsigma^4)} 
 \end{equation}
 with $\arg p_\perp
 (0)=\pi/2$.
 As is illustrated in Fig.~\ref{fig:Path2}, curves $\mathcal{C}_1$ and $\mathcal{C}_3$  start at $\varsigma_1\equiv k_1(0)-1/2=-i/2$ and $\varsigma_3\equiv k_3(\infty)-1/2$, respectively, and  terminate both at $\varsigma=-1/2$.  Curve $\mathcal{C}_2$ starts from $\varsigma\equiv k_2(0)-1/2$ and runs towards $1-\infty i$.
 \begin{figure}[htbp]
   \centering
 \includegraphics[width=.6\columnwidth]{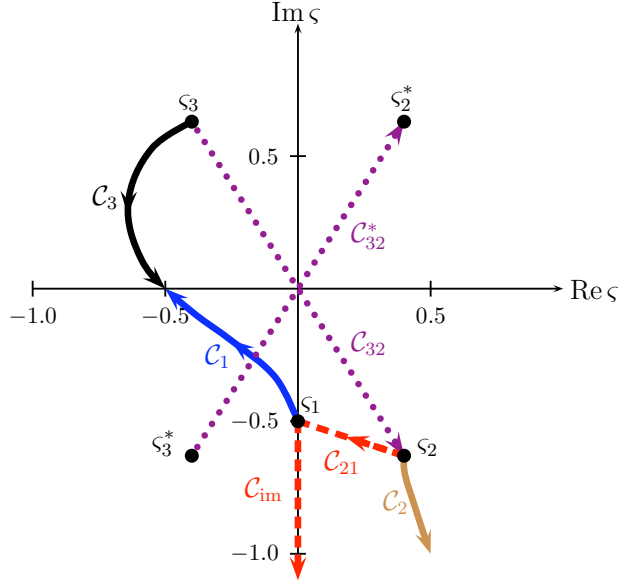}
   \caption{Integration paths used in Eqs.~\eqref{eq:IP} and \eqref{eq:IP32}. The location of the points $\varsigma_1$, $\varsigma_2,\dotsc, \varsigma_3^*$ displayed in the figure corresponds to the choice $P=1$. For further explanation, see main text.}
   \label{fig:Path2}
 \end{figure}
 Let us deform the path $C_2$ into the union of paths $\mathcal{C}_{21}+\mathcal{C}_{\text{im}}$ also displayed in Fig.~\ref{fig:Path2}. Since $\arg p_\perp(\varsigma)=0$ on $\mathcal{C}_{\text{im}}$, the contribution from $\mathcal{C}_{\text{im}}$ to the integral $\int_{\mathcal{-C}_2}d\varsigma$ in Eq.~\eqref{eq:IP} is purely imaginary. Thus it does not contribute to the real part of $ \mathcal{I}(P^2)$ and hence not to $\mathcal{J}(P^2)$ and $J(P)$ [Eqs.~\eqref{eq:JcalJ} and \eqref{eq:JcalI}]. The union of the remaining paths $\mathcal{C}_3-\mathcal{C}_1-\mathcal{C}_{21}$ can be deformed into a single path $C_{32}$ drawn violet and dotted in Fig.~\ref{fig:Path2}. Complex conjugation of the integral $\int_{\mathcal{C}_{32}}$ gives an integral along the complex conjugate path $\mathcal{C}_{{32}}^*$, which starts at the complex conjugate $\varsigma_3^*$ of $\varsigma_3$ and terminates at $\varsigma_2^*$. We thus arrive at the result
\begin{equation}
 \label{eq:IP32}
 J(P)=\frac{S_{d-3}}{8P(2\pi)^{d-2}\cos (\pi d) } \bigg\{ \int_{\mathcal{C}_{32}} d \varsigma \, [p_\perp (\varsigma)]^{d-4}+ e^{-i\pi d}\,\int_{\mathcal{C}^*_{32}} d \varsigma \, [p_\perp (\varsigma)]^{d-4}\bigg\}, 
 \end{equation}
where again $ \arg p_\perp (0)=\pi/2$. Finally, we transform from $\varsigma$ to the integration variable $t=-4 \varsigma^2/w$ with $w=P^2$.

The integral representation~\eqref{eq:IP32} is equivalent to the one given by Eqs.~\eqref{eq:JcalJ}--\eqref{eq:Jcal3}.  We shall prove this for  values of $P$ ($\simeq 1$) that are sufficiently small so that the absolute value $r\equiv |\varsigma_2|=|\varsigma_3|$ satisfies $r>|\varsigma_1|$.  The generalization to the half-axis $0<P<\infty$ then follows by analytic continuation in $P$.

Let us deform the integration path $\mathcal{C}_{32}$  of Fig.~\ref{fig:Path2}  into the full violet curve depicted in Fig.~\ref{fig:Paths3}, and  $\mathcal{C}_{32}^*$ likewise  into the dotted brown curve. 
\begin{figure}[htbp]
\begin{center}
\includegraphics[width=0.5\textwidth]{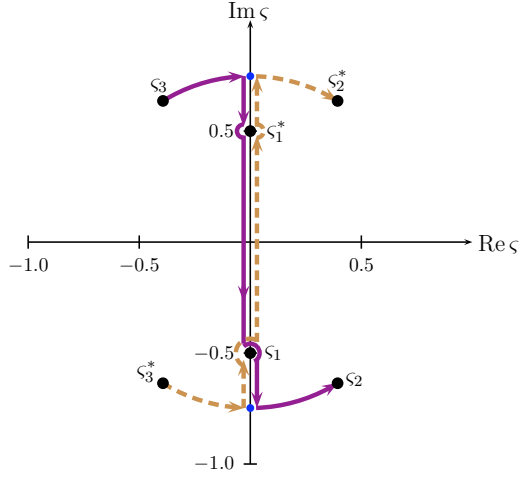}
\caption{Deformation of the integration paths $\mathcal{C}_{32}$ and $\mathcal{C}_{32}^*$ into the full violet and dashed brown curves.}
\label{fig:Paths3}
\end{center}
\end{figure}
Our choice of the phase of $p_\perp(\varsigma)$ such that $\arg p_\perp(ir-0)=\pi/2$ implies that the phase relations $\arg p_\perp(ir-0)=\arg p_\perp(-ir+0)=0$ and  $\arg p_\perp(ir-0)=\arg p_\perp(-ir+0)=0$ hold on the paths $\mathcal{C}_{32}$ and $\mathcal{C}_{32}^*$, respectively. As a consequence, the sum of the contributions from the two integrals between $ir\mp0$ and $\varsigma_1^*\mp 0$   cancel. Likewise, the integrals along the paths between  $\varsigma_1\mp0$ and $-ir\mp 0$ add up to zero.
The contributions from the remaining portions of the paths in Fig.~~\ref{fig:Path2} add up to
\begin{eqnarray}
 \label{eq:IP32a}
 J(P)&=& \frac{S_{d-3}}{8P(2\pi)^{d-2}\cos (\pi d) } \bigg\{ \int_{\varsigma_3}^{\varsigma_2^*} d \varsigma \, 
[p_\perp (\varsigma)]^{d-4} +
\int_{\varsigma_3^*}^{\varsigma_2} d \varsigma \, 
[p_\perp (\varsigma)]^{d-4} \nonumber\\ &&\strut+ 
(1-e^{-i\pi d})\,\int_{\varsigma_1^*}^{\varsigma_1} d \varsigma \, [p_\perp (\varsigma)]^{d-4}\bigg\}, 
 \end{eqnarray} 
where phases of $p_\perp (\varsigma)$ in the three integrals are fixed by the conditions \begin{equation} \label{eq:arg2}
\arg p_\perp (i r)=\arg p_\perp (-i r)=0, \quad
\arg p_\perp (0)=\pi/2. 
\end{equation}
Further, the principal value of the power $ [p_\perp(\varsigma)]^{d-4} $ is to be taken, i.e.\
$\arg [p_\perp(\varsigma)]^{d-4}=(d-4)\arg p_\perp(\varsigma)$. Since the integrands of these integrals are even, Eq.~\eqref{eq:IP32a} can be rewritten as
\begin{equation}
 \label{eq:IP32b}
 J(P)= \frac{  S_{d-3}}{4 P(2\pi)^{d-2}\cos (\pi d) } \bigg\{ \int_{\varsigma_3^*}^{\varsigma_2} d \varsigma \, [p_\perp (\varsigma)]^{d-4}+ (1-e^{-i\pi d})\,\int_{0}^{\varsigma_1} d \varsigma \, [p_\perp (\varsigma)]^{d-4}\bigg\}.
 \end{equation} 
Both integration paths lie in the lower half-plane $\Im\,\varsigma<0$, and the phase
conditions \eqref{eq:arg2} must be taken into account. Changing to  the integration variable $t=-4 \varsigma^2/w$ with $w=P^2$ finally yields the representation given in  Eqs.~\eqref{eq:JcalJ}--\eqref{eq:Jcal3}.
 
 \section{Branching of complex integrals ${\mathcal J}_k (w)$ \label{ap:mon}}

The functions $\mathcal{J}_1(w)$, $\mathcal{J}_3(w)$, and $\mathcal{J}(w)$ we considered in Sec.~\ref{sec:endens} for $0<w<\infty$ can be analytically continued to the complex $w$-plane. These analytic continuations become multivalued functions with four branch points at $w_1=-1$, $w_2= -1/4$, $w_3= 0$, and $w_4= \infty$. The branching of these functions, which we are now going to study, is  described by the monodromy group.

To this end, consider a real value $w_0$ of $w$ with $-1/4<w_0<0$  and $\arg\,w_0=\pi$. For such
 $w_0$, all branch points of the integrands of the integrals in terms of which the functions $\mathcal{J}_1(w)$, $\mathcal{J}_3$, and $\mathcal{J}(w)$ were expressed in Eqs.~\eqref{eq:calJw}--\eqref{eq:Jcal3} are real and given by 
 \begin{equation}
  t_1(w_0)\equiv t_-(w_0)\,,\;\;t_2(w_0)\equiv t_+(w_0)\,,\;\;t_3\equiv 0\,,\;\;t_4\equiv 1\,, 
\end{equation}
 where $t_\mp(w)$ are the zeros~\eqref{eq:tpm} of the function $\rho(t,w)$ introduced in Eq.~\eqref{eq:rho}. They 
 satisfy $t_1(w_0)<t_2(w_0)<t_3<t_4$.

Let $ {\mathcal J}_k (w_-)$, with $k=1,2,3$, denote the integrals
\begin{equation}
  \label{eq:defJks}
 {\mathcal J}_k (w_-)= w_-^{2 \lambda}  \int_{t_k}^{t_{k+1}} dt \,
  [t-t_1 (w_-)]^\lambda \,[t-t_2 (w_-)]^\lambda 
   \,(t-t_3)^{-1/2} \,(t-t_4)^\lambda.
\end{equation}
For $k=1$ and $3$, these definitions comply with Eqs.~\eqref{eq:Jcal1} and \eqref{eq:Jcal3}.
 The integrand $t^{-1/2}\prod_{k=1,2,4}(t-t_k)^\lambda$  is a multivalued function. Following Ref.~\cite{Mas97}, we fix its phase by setting 
 \begin{equation} \label{eq:arg}
  \arg (t-t_j) = \begin{cases}
   0 & \text{if }j \leq k,\\[\medskipamount]
 -\pi & \text{if }k+1 \leq j.
  \end{cases} 
 \end{equation}
This  guarantees that the integrand is an analytical function in the \emph{lower} half-plane $\Im\,t <0$.

The functions $\mathcal{J}_k(w_-)$ can be analytically continued from the interval $(-1/4,0)$ into the complex 
$w$-plane punctured at the branch points $w_i=-1,-1/4,0$. At these branch points, some of the endpoints 
$t_k(w)$ of the integration paths in Eq.~\eqref{eq:defJks} merge or become infinite. Namely, $t_2(w)\to t_3$ as $w\to-1/4$,  $t_{1,2}(w)\to \infty $ as $w\to 0$, and $t_1(w)\to t_3$ as $w\to-1$. To study the branching of the integrals ${\mathcal J}_k(w)$ at the points $w_1$, $w_2$, and $w_3$, let us  change $w$ continuously by moving along loops $\gamma_i$ that emanate from $w_-$ and terminate there, passing counter-clockwise around one of the branch points $w_i$ of the functions $\mathcal{J}_k(w)$, as is  illustrated in Fig.~\ref{fig:wloops}.
\begin{figure}[htbp]
 \begin{center}
 \includegraphics[width=.8\columnwidth]{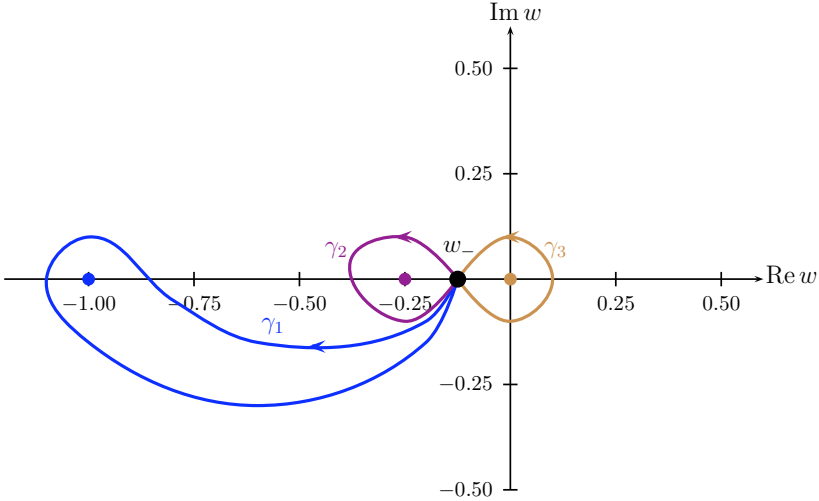}
 \end{center}
  \caption{The branching of the integrals ${\mathcal J}_k (w_0)$, $k=1,2,3$, is characterized completely by the monodromy transformations  $\hat{A}(\gamma_i)$ pertaining to the loops $\gamma_1$, $\gamma_2$, and $\gamma_3$.}
  \label{fig:wloops}
\end{figure}

As $w$ is changed continuously,  the  functions $\mathcal{J}_k (w)$ also change continuously. However, because of the nontrivial monodromy, they do not normally return to the starting values $\mathcal{J}(w_-)$ if the loop is traversed a single time. Let $\hat{A} (\gamma_i)\,\mathcal {J}_k (w_-)$  denote the end value one reaches from $\mathcal {J}_k (w_-)$ by going once around the loop $\gamma_i$. Although  $\hat{A} (\gamma_i)\,\mathcal {J}_k (w_-)$  generally differs from $\mathcal {J}_k (w_-)$, it must be a linear combination of the three linearly independent solutions of the differential equation~\eqref{eq:dif} it solves. One finds
\begin{equation}
  \hat{A} (\gamma_1) \begin{pmatrix}\mathcal{J}_1(w_-)\\ \mathcal{J}_2(w_-)\\ \mathcal{J}_3(w_-)\end{pmatrix} = \begin{pmatrix}1&-b_{\lambda}&b_{\lambda}\\ 0&1+b_{\lambda}&-b_{\lambda}\\ 0&-b_{\lambda}c_{\lambda}&1+b_{\lambda}c_{\lambda}
  \end{pmatrix}
  \begin{pmatrix}\mathcal{J}_1(w_-)\\ \mathcal{J}_2(w_-)\\ \mathcal{J}_3(w_-)\end{pmatrix},
  \end{equation}
  \begin{equation}
  \hat{A} (\gamma_2) \begin{pmatrix}\mathcal{J}_1(w_-)\\ \mathcal{J}_2(w_-)\\ \mathcal{J}_3(w_-)\end{pmatrix} = \begin{pmatrix}1&2&0\\ -1&-c_{\lambda}^{-1}&1-c_{\lambda}^{-1}\\ 0&0&1
  \end{pmatrix}
  \begin{pmatrix}\mathcal{J}_1(w_-)\\ \mathcal{J}_2(w_-)\\ \mathcal{J}_3(w_-)\end{pmatrix},
  \end{equation}
  and
   \begin{equation}
  \hat{A} (\gamma_3) \begin{pmatrix}\mathcal{J}_1(w_-)\\ \mathcal{J}_2(w_-)\\ \mathcal{J}_3(w_-)\end{pmatrix} = \begin{pmatrix}1&0&0\\ 0&-c_{\lambda}&0\\ 0&c_{\lambda}-1&1
  \end{pmatrix}
  \begin{pmatrix}\mathcal{J}_1(w_-)\\ \mathcal{J}_2(w_-)\\ \mathcal{J}_3(w_-),\end{pmatrix}
  \end{equation}
with
\begin{equation}
c_\lambda\equiv e^{2\pi i \lambda }\quad\text{and}\quad b_\lambda\equiv(c_\lambda-1)(c_\lambda^2+1).
\end{equation}

These equations describe the action of the monodromy group on the three basic integrals  ${\mathcal J}_k$, $k=1,2,3$. As an immediate consequence, we obtain for the monodromy group action on the integral
 ${\mathcal J}(w_-)$ the result
 \begin{eqnarray} \label{eq:monc}
  \hat{A} (\gamma_1) {\mathcal J}(w_-) &=& {\mathcal J}(w_-)+i \,b_\lambda\, c_\lambda \,[{\mathcal J}_2(w_-) -{\mathcal J}_3(w_-)],\nonumber\\
  \hat{A} (\gamma_2) {\mathcal J}(w_-) &=& {\mathcal J}(w_-) -i \big[c_\lambda+c_\lambda^{-1}\big]{\mathcal J}_2(w_-),\nonumber  \\
  \hat{A} (\gamma_3) {\mathcal J}(w_-) &=& {\mathcal J}(w_-). 
 \end{eqnarray}
Directly at the upper critical dimension $d^*(1)=9/2$, one has $\lambda=(9/2-4)/2=1/4$, $c_\lambda=i$, and $b_\lambda=0$, as a consequence of which $\hat A (\gamma_k)\mathcal{J}(w_-) =\mathcal{J}(w_-) $ for $k=1,2,3$.

 %\bibliography{bank,./remarks}

 \end{document}